\documentclass[journal]{IEEEtran}

\makeatletter
\usepackage{fixltx2e}
\usepackage{xspace}
\usepackage[cmex10]{amsmath}
\usepackage{amsfonts}
\usepackage{amssymb}
\usepackage[]{hyperref}
\usepackage[all]{hypcap}
\usepackage{amsthm}
\usepackage{url}
\usepackage{ifthen}

\usepackage{array}
\usepackage{color}
\usepackage[absolute,overlay]{textpos}
\usepackage{tikz}
\tikzstyle{block}=[draw opacity=0.7,line width=1.4cm]
\usetikzlibrary{arrows,decorations,positioning,backgrounds,fit,shapes,matrix}

\def\push{\texttt{PUSH}\xspace}
\def\pull{\texttt{PULL}\xspace}
\def\ex{\texttt{EXCHANGE}\xspace}
\def \exchange{\texttt{EXCHANGE}\xspace}
\def\proto{$\mathcal{P}$\xspace}
\def\algo{$\mathcal{A}$\xspace}
\def\P{\mathcal{P}\xspace}
\def\A{\mathcal{A}\xspace}

\newtheorem{theorem}{Theorem}
\newtheorem*{theorem_jackson}{Jackson's Theorem}
\newtheorem*{theorem_main}{Theorem 1}
\newtheorem{lemma}{Lemma}
\newtheorem{corollary}{Corollary}
\newtheorem{dfn}{Definition}
\newtheorem{clm}{Claim}

\newcommand{\rmnum}[1]{\romannumeral #1}
\newcommand{\Rmnum}[1]{\expandafter\@slowromancap\romannumeral #1@}

\pdfpageattr{/Group <</S /Transparency /I true /CS /DeviceRGB>>}

\begin{document}

\title{Bounds for Algebraic Gossip on Graphs}

\author{
\IEEEauthorblockN{ Michael Borokhovich \;\;\; Chen Avin \;\;\;  Zvi Lotker}

\thanks{Authors are with the Department of
Communication Systems Engineering, Ben-Gurion University of the Negev, Israel, e-mails:\{borokhom, avin, zvilo\}@cse.bgu.ac.il.}
}


\maketitle


\begin{abstract}
We study the stopping times of gossip algorithms for network coding. We analyze algebraic gossip (i.e., random linear coding) and consider three gossip algorithms for information spreading: Pull, Push, and Exchange. 
The stopping time of algebraic gossip is known to be linear for the complete graph, but the question of determining a tight upper bound or lower bounds for general graphs is still open. 
We take a major step in solving this question, and prove that algebraic gossip on any graph of size $n$ is $O(\Delta n)$ where $\Delta$ is the maximum degree of the graph. This leads to a tight bound of $\Theta(n)$ for bounded degree graphs and an upper bound of $O(n^2)$ for general graphs. We show that the latter bound is tight by providing an example of a graph with a stopping time of $\Omega(n^2)$.
Our proofs use a novel method that relies on Jackson's queuing theorem to analyze the stopping time of network coding; this technique is likely to become useful for future research.
\end{abstract}

\begin{keywords}
Gossip, Algebraic Gossip, Network Coding, Gossip Algorithms,
Network Capacity.
\end{keywords}


\section{Introduction}\label{sec:introduction}

Randomized gossip-based protocols are attractive due to their locality, simplicity, and structure-free nature, and have been offered in the literature for various tasks, such as ensuring database consistency and computing aggregate information \cite{KarpEtAl00a,Kempe2003Gossip,Boyd2006Randomized}.
Consider the case of a connected network with $n$ nodes, each holding a value it would like to share with the rest of the network. Motivated by wireless networks and limited resource sensor motes, in recent years researchers have studied the use of randomized gossip algorithms together with network coding for this multicast task \cite{Medard2002Beyond,LiYeuCai03}. The use of network coding protocols for multicast has received growing attention due to the ability of such protocols to significantly increase network capacity. For a basic network coding example, see \cite{DBLP:journals/ccr/FragouliBW06}.

In this work we consider \emph{algebraic gossip}, a gossip-based network coding protocol known as random linear coding \cite{Ho2003The-benefits}.
In the discussion on gossip-based protocols we distinguish between the \emph{gossip algorithm} and the \emph{gossip protocol}.
A gossip algorithm is a communication scheme in which at every timeslot, a random node chooses a random neighbor to communicate with. We consider three known gossip algorithms: \push: a message is sent \emph{to} the neighbor, \pull: a message is sent \emph{from} the chosen neighbor, and \ex: the two nodes exchange messages. The gossip protocol, on the other hand, determines the \emph{content} of messages sent. In algebraic gossip protocol, the content of the messages is a random linear combination of all messages stored by the sending node. Once a node has received enough independent messages (independent linear equations) it can solve the system of linear equations and discover all the initial values of all other nodes. 


We study the performance of algebraic gossip on arbitrary network topologies, where information is disseminated from all nodes in the network to all nodes, i.e., all-to-all dissemination. Previously, algebraic gossip was considered with \push and \pull gossip algorithms; here we also study the use of \ex, which can lead to significant improvements for certain topologies (as we show).
Our main goal is to find tight bounds for the stopping time of the algebraic gossip protocol, both in expectation and with high probability (w.h.p.), i.e., with probability of at least $1-\frac{1}{n}$. 

The stopping time question, i.e., bounding the number of rounds until protocol completeness, has been addressed in the past. Deb \emph{et al.} \cite{Deb2006Algebraic} studied algebraic gossip using \pull and \push on the complete graph and showed a tight bound of $\Theta(n)$, a linear stopping time, both in expectation and with high probability.
Boyd \emph{et al.} \cite{Boyd2006Randomized, Boyd2005Gossip} studied the stopping time of a gossip protocol for the \emph{averaging problem} using the \ex algorithm. They gave a bound for symmetric networks that is based on the second largest eigenvalue of the transition matrix or, equally, the mixing time of a random walk on the network, and showed that the mixing time captures the behavior of the protocol. Mosk-Aoyama and Shah \cite{Mosk-Aoyama2006Information} used a similar approach to \cite{Boyd2006Randomized} and \cite{Boyd2005Gossip} to analyze algebraic gossip on arbitrary networks. They consider symmetric stochastic matrices that (may) lead to a non-uniform gossip and gave an upper bound for the \pull algorithm that is based on a measure of conductance of the network. As the authors mentioned, the offered bound is not tight, which indicates that the conductance-based measure does not capture the behavior of the protocol.

A recent independent work by Vasudevan and Kudekar \cite{vasudevan-2009} also offered the use of \ex together with algebraic gossip. Moreover, the authors give a uniform strong bound on algebraic gossip for arbitrary networks: $O(n\log n)$ in expectation and  $O(n\log^2n)$ with high probability.
The question about a worst case topology for algebraic gossip was not previously addressed in the literature.


\subsection*{Overview of Our Results}

The main contribution of this paper is new bounds for the stopping time of algebraic gossip on arbitrary graphs.\footnote{The short version of this paper appeared in the Proceedings ISIT 2010 \cite{Borokhovich2010Tight}.} Our bounds are tight for many graph families; moreover, for almost any chosen maximum degree there exist graphs for which the bounds are tight.
First, in Theorem \ref{thm:lower_bound_any_graph} we disprove the results of \cite{vasudevan-2009} by providing a graph for which algebraic gossip takes $\Omega(n^2)$ rounds. 
Our main result then, Theorem \ref{thm:upper_bound_any_graph}, gives an upper bound of $O(\Delta n)$ for the stopping time of algebraic gossip on any graph, where $\Delta$ is the maximum degree in the graph. 

\begin{theorem}
\label{thm:upper_bound_any_graph}
For the asynchronous (synchronous) time model and for any graph $G$ of size $n$ with maximum degree $\Delta$, the stopping time of algebraic gossip is $O(\Delta n)$ rounds both in expectation and with high probability.
\end{theorem}

This result immediately leads to two interesting corollaries.
In Corollary \ref{cor:upper} we state a matching upper bound: for any graph of size $n$, since the max $\Delta= n$, algebraic gossip will stop w.h.p. in $O(n^2)$ rounds. In Corollary \ref{corollary_any_ex_nc_linear} we conclude a strong tight bound for \emph{any} constant degree network (i.e., $\Delta$ is constant) of $\Theta(n)$.
This improves upon known previous upper bounds that, for certain constant degree graphs, had an upper bound of $O(n^2)$.
Note that the bound of $O(\Delta n)$ is not tight for all graphs (e.g., the complete graph) and 
the question of determining the properties of a network that capture tightly the stopping time of algebraic gossip
is still open.
We also show in Theorem \ref{thm:lower_bound_any_graph} that the upper bound $O(\Delta n)$ is tight in the sense that for almost any $\Delta$ there exist graphs for which algebraic gossip takes $\Omega(\Delta n)$ rounds.

The second contribution of the paper is the technique we use to prove our results. 
We novelly 
bound the stopping time of algebraic gossip via reduction to a network of queues and by the \emph{stationary} state of the network 
that follows form \emph{Jackson's theorem} for an open network of queues.
The idea of using a queuing theory approach for network coding analysis was first introduced in \cite{Lun20083} but, as opposed to our approach, it did not include the aspect of a gossip communication model.
We believe that the type of reduction presented in this work could be used for future analysis of gossip protocols. 


Third, we compare three gossip algorithms: \push, \pull, and \ex.
While traditionally algebraic gossip used \pull or \push as its gossip algorithms, it was unclear if using \exchange (that uses twice as many messages than \pull or \push) can lead to significant improvements in stopping time. We give a surprising affirmative answer to this question and prove that, for some topologies, using the \exchange algorithm can be unboundedly better than using \pull or \push. We show that while the time it takes the \ex algorithm to complete the algebraic gossip on the star graph is linear, i.e., $O(n)$ the time it takes the \pull and \push algorithms to finish the same task, is $\Omega(n\log n)$. On the contrary, there are many other graphs such as the complete graph and all constant maximum degree graphs (see Section \ref{sec:bound_for_arbitrary}), on which these three gossip algorithms have the same asymptotical behavior. 

Since the submission of this manuscript, there have been some recent advances in answering open questions raised in this work. In particular, the conference paper of Haeupler \cite{Haeupler2010Analyzing} and our recent conference paper \cite{Avin2011OrderOptimal}. We discuss these works in Conclusions.

The rest of the paper is organized as follows. In Section \ref{sec:preliminaries_and_models} we present the communication and time models, define gossip algorithms and gossip protocols, and formally state the \emph{gossip stopping problem}. In Section \ref{sec:ring_is_linear} we show that algebraic gossip on the ring graph is linear using Jackson's theorem. In Section \ref{sec:bound_for_arbitrary} we prove our main results: a tight upper bound for arbitrary networks and a tight linear bound for graphs with a constant maximum degree. Section \ref{sec:ex_is_faster} gives an answer to the question: ``Can \ex be better than \push or \pull?'' by providing a topology for which \ex is unboundedly faster. We conclude in Section \ref{sec:conclusions}.

\section{Preliminaries and Models}\label{sec:preliminaries_and_models}

\subsection{Network and Time Model}

We model the communication network by a connected undirected graph $G = G(V,E)$, where $V=\left\{v_1,v_2,...,v_n\right\}$ is the set of vertices and $E\subseteq{V\times V}$ is the set of edges. Let $N(v)\subseteq V$ be a set of neighbors of node $v$ and $d_{v}=\lvert N(v)\rvert$ its degree, let $\Delta = \max_v d_v$ be the maximum degree of $G$.

The time is assumed to be slotted where $n$ consecutive timeslots are regarded as one \emph{round}.
We consider the following time models:

\begin{itemize}

\item\textbf{Asynchronous time model.} At every timeslot, a node selected independently and uniformly at random takes an action (determined by a Gossip Algorithm) and a single pair of nodes communicates.\footnote{Alternately, this model can be seen as each node having a clock that ticks at the times of a rate 1 Poisson process and there is a total of $n$ clock ticks per round \cite{Boyd2006Randomized}.} In this model there is no guarantee that a node will be selected exactly once in a round; nodes can be selected several times or not at all.

\item\textbf{Synchronous time model.} 
At every round, all the nodes wake up synchronously and every node takes an action (determined by a Gossip Algorithm).
\end{itemize}

\subsection{Gossip Algorithms}\label{sec:algo}

A \emph{gossip algorithm} defines the way information is exchanged or spread in the network. When a node wakes up (according to a time model), it takes an information spreading action that is divided into two phases: (\rmnum{1}) choosing a communication partner and (\rmnum{2}) spreading the information. A \emph{communication partner} $u \in{N(v)}$ is chosen by node $v \in{V}$ with probability $p_{vu}$. Throughout this paper we will assume \textit{uniform gossip algorithms}, i.e., $p_{vu}=\frac{1}{d_v}$. 

We distinguish three gossip algorithms for information spreading between $v$ and $u$, $\push, \pull$, and $\ex$ as explained in the Introduction.
We assume that in the asynchronous time model, messages sent in timeslot $t$ are received in timeslot $t$ and can be forwarded or processed at timeslot $t+1$, and in the synchronous time model, messages sent in round $t$ are received in round $t$ and can be forwarded or processed at round $t+1$.

\subsection{Algebraic Gossip}
\label{sec:proto}

A \emph{gossip protocol} is a task that is being executed using gossip algorithms, for example, calculation of aggregate functions, resource discovery, and database consistency. 
We now describe the algebraic gossip protocol for the multicast task: disseminating $n$ initial values of the nodes to all $n$ nodes.

Let $\mathbb{F}_q$ be a field of size $q$, each node $v_i\in V$ holds an initial value $x_i$ that is represented as a vector in $\mathbb{F}_q^r$. We can represent every message as an integer value bounded by $M$, and therefore, $r=\left\lceil \log_q(M)\right\rceil$. All transmitted messages have a fixed length and represent linear equations over $\mathbb{F}_q$. The variables of these equations are the initial values $x_i\in\mathbb{F}_q^r$ and a message contains the coefficients of the variables and the result of the equation; therefore the length of each message is: $r\log_2q+n\log_2q$ bits. A message is built as a random linear combination of all messages stored by the node and the coefficients are drawn uniformly at random from $\mathbb{F}_q$.
A received message will be appended to the node's stored messages only if it is independent of all linear equations (messages) that are already stored by the node and otherwise it is ignored.
Initially, node $v_i$ has only one linear equation that consists of only one variable corresponding to $x_i$ multiplied by a coefficient $1$ and equal to the value of $x_i$, i.e., the node knows only its initial value.
Once a node receives $n$ independent equations it is able to decode all the initial values and thus completes the task.

For a node $v$ at timeslot (round)\footnote{For asynchronous time model -- timeslot, for synchronous -- round.} $t$, let $S_{v}(t)$ be the subspace spanned by the linear equations (or vectors) it stores (i.e., the coordinates of each vector are the coefficients of the equation) at the beginning of timeslot (round) $t$. The dimension (or rank) of a node is the dimension of its subspace, i.e., $dim(S_v(t))$ and it is equal to the number of independent linear equations stored by the node.

We say that a node $v$ is a \textbf{\emph{helpful node}} to node $u$ at the timeslot (round) $t$ if and only if $S_v(t)\not\subset S_u(t)$, i.e., iff a random linear combination constructed by $v$ can be linearly independent with all equations (messages) stored by $u$.
We call a message a \textbf{\emph{helpful message}} if it increases the dimension of the node. 
The following lemma, which is a part of Lemma 2.1 in \cite{Deb2006Algebraic}, gives a lower bound for the probability of a message sent by a \emph{helpful node} to be a \emph{helpful message}.
\begin{lemma}[\cite{Deb2006Algebraic}]
\label{lemma_helpful_rank}
Suppose that node $v$ is \emph{helpful} to node $u$ at the beginning of the timeslot (round) $t$. If $v$ transmits a message to $u$ at the timeslot (round) $t$, then:
$$\Pr \left(dim(S_u(t+1))>dim(S_u(t))\right) \geq 1-\tfrac{1}{q}.$$
That is, the probability of the message to be helpful is at least $1-\frac{1}{q}$.
\end{lemma}

\subsection{The Gossip Stopping Problem}\label{sec:stop}
Our goal is to compute bounds on time and number of messages needed to be sent in the network to complete various gossip protocols over various gossip algorithms. For this purpose we define the following:

\begin{dfn}[Excepted and high probability stopping times]
Given a graph $G$, gossip algorithm \algo, and a gossip protocol \proto, the stopping time $T(\mathcal{A}, \mathcal{P}, G)$ is a random variable defined as the number of timeslots by which all nodes complete the task. 
$E[T(\mathcal{A}, \mathcal{P}, G)]$ is the expected stopping time and 
the \emph{high probability stopping time} $\hat{T}$ is defined as follows:
$$\hat{T}(\mathcal{A}, \mathcal{P}, G) = \min_{t\in\mathbb{Z}} \left[t\mid\Pr \left (T(\mathcal{A}, \mathcal{P}, G) \leq t \right ) \geq 1- \tfrac{1}{n}\right].$$
\end{dfn}


We can now express our research question formally:
\begin{dfn}[Gossip stopping problem]
Given a graph $G$, a gossip algorithm $\A$, and a gossip protocol $\P$, the \emph{gossip stopping problem} is to determine $E[T(\mathcal{A}, \mathcal{P}, G)]$ and $\hat{T}(\mathcal{A}, \mathcal{P}, G)$, the expected and  high probability stopping times.
\end{dfn}

In this work we consider $\mathcal{A}\in \{\push, \pull,\ex\}$ and $\mathcal{P}=$ \emph{algebraic gossip}, so when these parameters and $G$ are understood from the context, we denote the expected and high probability stopping times as $E[T]$ and $\hat{T}$, respectively. Moreover, we usually measure the stopping time in \emph{rounds} (in order to compare between the two time models) where one round equals $n$ consecutive timeslots. Thus, we define
the expected number of rounds as $E[R] = E[T]/n$ and 
 $\hat{R}=\hat{T}/n$ as the number of rounds by which all nodes complete the task with high probability.

For clarity, we present our proofs for the \emph{asynchronous} time model and the $\ex$ algorithm, we extend the results to the \emph{synchronous} cases and \push and \pull in the appendix, where we also included some of the more technical proofs.

\section{Linear Bound on a ring via Queuing Theory}
\label{sec:ring_is_linear}

Before proving the main results of Theorem \ref{thm:upper_bound_any_graph} in the next section we prove in this section
a bound on the specific case of a ring network. This is a simpler case to prove and understand, and will be used as a basis for the proof of the general result. A ring of size $n$ is a connected cycle where each node has one left and one right neighbor.
 
\begin{theorem} 
\label{thm:ring_is_linear}
For the asynchronous time model and the ring graph of size $n$, the stopping time (measured in rounds) of algebraic gossip is linear both in expectation and with high probability, i.e., 
$\text{E}\left[R\right]=\Theta(n)$ and $\hat{R}=\Theta(n)$.
\end{theorem}

\begin{figure*}
\centering
\includegraphics[width=4in,clip=true, viewport=0.1in 2.3in 11.5in 6.1in]{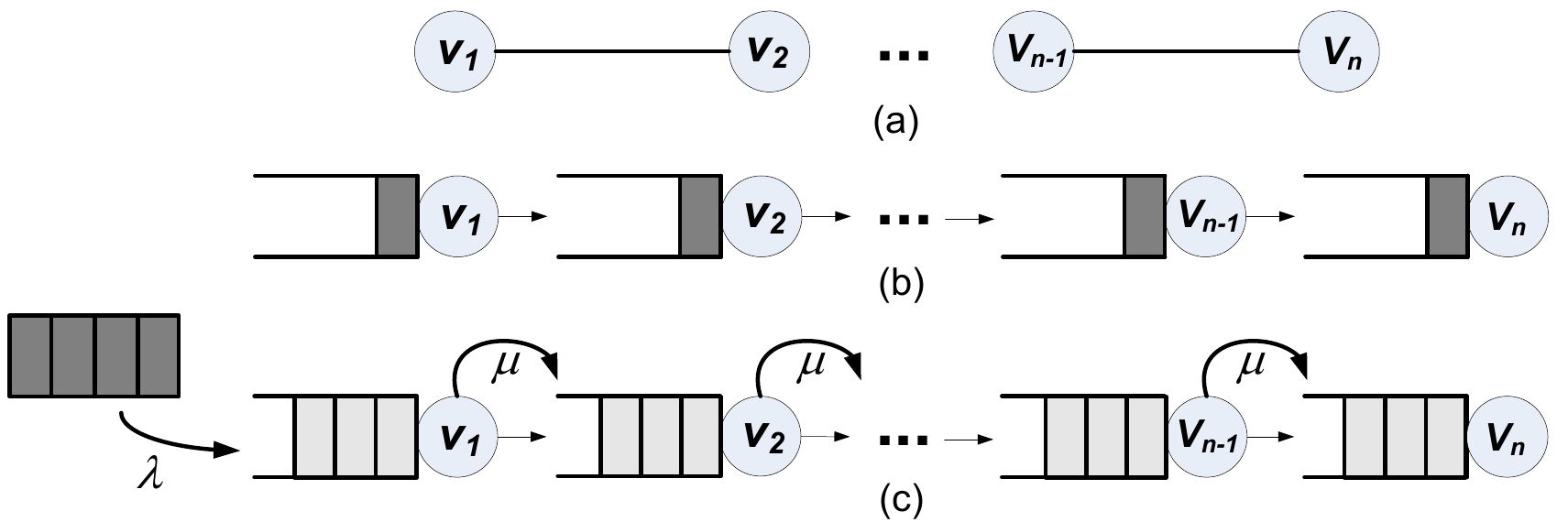}
\caption{Modeling algebraic gossip in a path as a queuing network. (a)Initial path graph. (b)One real customer at each node. (c)Queues are filled with dummy customers and real customers enter the system from outside.
\label{fig:fig/queues.pdf}}
\end{figure*}

\begin{proof}
The idea of the proof is to reduce the problem of network coding on the ring graph to a simple system of queues and use Jackson's theorem for open networks to bound the time it takes \emph{helpful messages} to cross the network.

To simplify our analysis, we cut the ring in an arbitrary place and get a path graph (without loss of generality, we assume that the leftmost node in the path is $v_1$ and the rightmost node is $v_n$), see Fig. \ref{fig:fig/queues.pdf} (a). It is clear that the stopping time of the algebraic gossip protocol will be larger in a path graph than in a ring graph. Another simplification that we will do, for the first part of the proof, is to consider only the messages that travel from left to right (towards $v_n$) (i.e., other messages will be ignored, thus increasing the stopping time).

We define a queuing system by assuming a queue with a single server at each node. Customers of our queuing network are the \emph{helpful messages}, i.e., messages that increase the rank of a node they arrive at. This means that every customer arriving at some node increases its rank by $1$, so the queue size at a node represents a measure of \emph{helpfulness} of the node to its right-hand neighbor (i.e., the queue size is the number of independent linear equations that the node can generate for its right-hand neighbor).
%
%
%
The service procedure at node $v_i$ is a transmission of a \emph{helpful message} (customer) from $v_i$ to $v_{i+1}$. 
%
%
So, from Lemma \ref{lemma_helpful_rank}, the probability that a customer will be serviced at node $v_i$ in a given timeslot is: $p\geq\frac{1}{n}(1-\frac{1}{q})$, where $\frac{2}{n}\cdot\frac{1}{2}=\frac{1}{n}$ is the probability that in the \ex algorithm a message will be sent from $v_i$ to $v_{i+1}$ at any given timeslot.

Thus, we can consider that a service time in our queuing system is geometrically distributed with parameter $p$. The service time is distributed over the set $\left\{0,1,2,...\right\}$, which means that a customer that enters an empty queue at the \emph{end} of the timeslot can be immediately serviced with probability $p$ (since it is the beginning of the next timeslot). A customer cannot pass more than one node (queue) in a single timeslot, so we define the transmission time as one timeslot. I.e., the time needed for a customer to pass through $k$ queues is the sum of the waiting time in each queue, service time in each queue, and additional $k$ timeslots for transmission from queue to queue.

The following lemma shows that the service rate can be bounded from below by an exponential random variable. The proof of this lemma can be found in Appendix \ref{app:D}.

\begin{lemma}\label{lemma:geometric_as_exponential}
Let $X$ be a geometric random variable with parameter $p$ and supported on the set $\left\{0,1,2,\ldots\right\}$, i.e., for $k\in\mathbb{Z^+}$: $\Pr \left(X=k\right)=(1-p)^{k}p$, and let $Y$ be an exponential random variable with parameter $p$. Then, for all $x\in \mathbb{R^+}$:
\begin{align}
\Pr \left(X\leq x \right)\geq \Pr \left(Y\leq x \right)=1-e^{-px},
\end{align}
i.e., a random variable $Y\sim\text{Exp}(p)$ stochastically dominates the random variable $X\sim\text{Geom}(p)$.
\end{lemma}

We can now assume that the service time is exponentially distributed with parameter $\mu=p$.
This assumption decreases the rate of transmission of \emph{helpful messages}, and therefore will only increase the stopping time. The last is true since the probability that a customer will be serviced by time $t_1$ in a geometrical server is higher than in an exponential server, and thus each customer in a network with geometric servers will arrive at $v_n$ by time $t_2$ with higher probability than in a network with exponential servers. The formal justification of this step is given later in Lemma \ref{lemma:exp_server_instead_of_geom}, which proves this assertion for trees and not only for the line.

To this end, we have converted our network to a standard network of queues where the network is open, external arrivals to nodes will form a Poisson process, service times are exponentially distributed, and the queues are first come first serve (FCFS).
For a queue $i$ let $\mu_i$ denote the service rate and $\lambda_i$ the total arrival rate. We present now Jackson's theorem for open networks; a proof of this theorem can be found in~\cite{Yao01}.

\begin{theorem_jackson}\label{thm:jackson}
In an open Jackson network of $n$ queues where the utilization $\rho_i=\frac{\lambda_i}{\mu_i}$ is less than 1 at every queue, the equilibrium state probability distribution exists, and for state $(k_1,k_2,\ldots,k_n)$ is given by the product of the individual queue equilibrium distributions:
$\pi(k_1,k_2,...,k_n)=\Pi_{i=1}^{n} \rho_i^{k_i}(1-\rho_i)$.
\end{theorem_jackson}

We would like to use Jackson's theorem to conclude that there is an equilibrium state for our network of queues and that in the equilibrium state the lengths of the queues are independent. For Jackson's theorem to hold we need to appropriately define the arrival rate to the queues, so we will slightly change our queuing network.

The initial state of our system is that at every queue we have one \emph{real} customer (see Fig. \ref{fig:fig/queues.pdf} (b)). Now we take all the $n$ \emph{real} customers out from the system and let them enter back via the leftmost queue with a predefined arrival rate. Clearly, this modification increases the stopping time. We define the \emph{real} customers' arrivals as a Poisson process with rate $\lambda= \frac{\mu}{2}$. So, $\rho_i=\frac{\lambda_i}{\mu_i}=\frac{1}{2}<1$ for all queues ($i\in \left[1..n\right]$).

Now, according to Jackson's theorem there exists an equilibrium state. So, our last step is to ensure that the lengths of all queues at time $t=0$ are according to the equilibrium state probability distribution. We add \emph{dummy}  customers to all the queues according to the stationary distribution. By adding additional \emph{dummy} customers (we call them \emph{dummy} since their arrivals are not counted as a rank increment) to the system, we make the \emph{real} customers wait longer in the queues, thus increasing the stopping time. Our queuing network with the above modifications is illustrated in Fig. \ref{fig:fig/queues.pdf} (c), where the \emph{real} customers are dark, and the \emph{dummy} customers are bright.

We will compute the stopping time in two phases.
By the end of the first phase, node $v_n$ will finish the algebraic gossip task. By the end of the second phase, all the nodes will finish the task.
For the first phase, we will find the time it takes the $n$'th (last) \emph{real} customer to arrive at the rightmost node, i.e., node $v_n$. By that time, the rank of node $v_n$ will become $n$ and it will finish the algebraic gossip protocol (i.e., it received $n$ \emph{helpful messages}). Let us denote this time (in \emph{timeslots}) as $T^{\overrightarrow{\text{arr}}}+T^{\overrightarrow{\text{cross}}}$, where $T^{\overrightarrow{\text{arr}}}$ is the time needed for the $n$'th customer to arrive at the first queue, and $T^{\overrightarrow{\text{cross}}}$ is the time needed for the $n$'th customer to pass through all the $n$ queues in the system.

For the second phase, let us assume that after $T^{\overrightarrow{\text{arr}}}+T^{\overrightarrow{\text{cross}}}$ timeslots (when $v_n$ finishes the algebraic gossip task) all nodes except node $v_n$ forget all the information they have. So, the rank of all nodes except $v_n$ is $0$. Let us now analyze the information flow from the rightmost node in the path ($v_n$) to the leftmost node ($v_1$). In the same way, we will represent all \emph{helpful messages} that node $v_n$ will send as customers in our queuing system. In order to use Jackson's Theorem, we will again remove all the \emph{real} customers from the system and will inject them to the queue of node $v_n$ with a Poisson rate $\lambda= \mu/2$. We also fill all the queues in the system with \emph{dummy} customers in order to achieve queue lengths that correspond to the equilibrium state distribution.
Clearly, arrival of a \emph{real} customer at some node $v_i$ ($i\neq n$) will increase the rank of that node. So, after the last \emph{real} customer arrives at node $v_1$, the ranks of \emph{all} nodes will be $n$, and the algebraic gossip task will be finished.

Using the same equilibrium state analysis as before, we define the time it takes the $n$'th (last) \emph{real} customer to arrive at the rightmost node $v_n$ as $T^{\overleftarrow{\text{arr}}}$, and the time to cross all the $n$ queues -- arriving at node $v_1$ -- as $T^{\overleftarrow{\text{cross}}}$.

So, $T=T^{\overrightarrow{\text{arr}}}+T^{\overrightarrow{\text{cross}}}+T^{\overleftarrow{\text{arr}}}+T^{\overleftarrow{\text{cross}}}$ is an upper bound for the number of timeslots needed to complete the task. Now we will find the upper bound for $T^x, x\in \left\{\overrightarrow{\text{arr}},\overrightarrow{\text{cross}},\overleftarrow{\text{arr}},\overleftarrow{\text{cross}}\right\}$ and then we will use union bound to obtain an upper bound on $T$.

From Jackson's Theorem, it follows that the number of customers in each queue is independent, which implies that the random variables that represent the waiting times in each queue are independent. 
To continue with the proof we need the following lemmas; the first is a classical result from queuing theory, the proof of the second lemma can be found in Appendix \ref{app:E}.

\begin{lemma}[\cite{1378238}, section 4.3]
\label{lemma:waiting_distr}
Time needed to cross one $M/M/1$ queue (a queuing system in which interarrival and service times are distributed exponentially with parameters $\lambda$ and $\mu$, respectively) in the equilibrium state has an exponential distribution with parameter $\mu -\lambda$.
\end{lemma}

\begin{lemma}\label{lemma:sum_of_exp_bounded1}
Let $Y$ be the sum of $n$ independent and identically distributed exponential random variables (each with parameter $\mu>0$) and $\text{E}\left[Y\right]=\tfrac{n}{\mu}$.
Then, for $\alpha>1$:
\begin{align}
\Pr \left(Y < \alpha\text{E}\left[Y\right]\right) > 1-(2e^{-\alpha/2})^n.
\end{align}
\end{lemma}

Recall that: 
$\mu=p \geq\frac{1}{n}(1-\frac{1}{q})$ so $\mu \ge \frac{q-1}{qn}\geq\frac{1}{2n}$ for $q\geq 2$.
The random variable $T^{\overrightarrow{\text{arr}}}$ is the sum of $n$ independent random variables distributed exponentially with parameter $\mu/2$. From Lemma \ref{lemma:waiting_distr} we obtain that $T^{\overrightarrow{\text{cross}}}$ is the sum of $n$ independent random variables distributed exponentially with parameter $\mu-\lambda=\tfrac{\mu}{2}$. It is clear that $T^{\overleftarrow{\text{arr}}}$ is distributed exactly as $T^{\overrightarrow{\text{arr}}}$ and $T^{\overleftarrow{\text{cross}}}$ is distributed exactly as $T^{\overrightarrow{\text{cross}}}$. Therefore (for $\mu=\frac{1}{2n}$): $\text{E}\left[T^x\right]=\sum_{i=1}^n\frac{2}{\mu}=4n^2$. Using Lemma \ref{lemma:sum_of_exp_bounded1} (with $\alpha=2$) we obtain for $x\in{ \left\{\overrightarrow{\text{arr}},\overrightarrow{\text{cross}},\overleftarrow{\text{arr}},\overleftarrow{\text{cross}}\right\}}$:
\begin{align}
\Pr \left(T^x\leq 8n^2\right)\geq 1-\left(\tfrac{2}{e}\right)^n.
\end{align}

Using a union bound we get that:
\begin{align}
\Pr \left(T\leq 32n^2\right) &\geq \Pr \left(\cap_{x} T^x\leq 8n^2\right)
\\&=1-\Pr \left(\cup_{x} T^x > 8n^2\right)\\
&\geq 1-4\left(\tfrac{2}{e}\right)^n.
\end{align}
It is clear that $\Pr \left(T\leq 32n^2\right)$ increases when $\mu$ increases (faster server yields smaller waiting time); hence, the above inequality holds for any $\mu\geq \tfrac{1}{2n}$.

So, for the asynchronous time model and \ex we obtain an upper bound for the high probability stopping time: $\hat{T}=O(n^2)$  in timeslots, and thus $\hat{R}=O(n)$, in rounds. 
Let us now find an upper bound for the expected number of rounds needed to complete the task -- $\text{E}\left[R\right]$:
\begin{align}
\text{E}\left[R\right]&=\tfrac{1}{n}\text{E}\left[T\right]
\\&=\tfrac{1}{n}\text{E}\left[T^{\overrightarrow{\text{arr}}}+T^{\overrightarrow{\text{cross}}}+T^{\overleftarrow{\text{arr}}}+T^{\overleftarrow{\text{cross}}}\right]
\\&=\tfrac{4}{n}\text{E}\left[T^x\right]=\tfrac{4}{n}4n^2=16n=O(n).
\end{align}
%
%

The lower bound is clear since in order to finish the algebraic gossip task each node has to receive at least $n$ messages, so at least $n^2$ messages need to be sent and received. Since in each timeslot at most 2 messages (using \ex) are sent, we get: $\hat{T}=\Omega(n^2)$, thus $\hat{R}=\Omega(n)$, and $\text{E}\left[R\right]=\Omega(n)$. 

The result of Theorem \ref{thm:ring_is_linear} is then follows: $E[R]=\Theta(n)$, and $\hat{R}=\Theta(n)$.
\end{proof}

\section{Algebraic Gossip on Arbitrary Graphs} \label{sec:bound_for_arbitrary}

Now we are ready to prove our main results. First, we present the upper bound for any graph as a function of its maximum degree $\Delta$, and then we give corollaries that are applications of this result for more specific cases. 

\begin{theorem_main}[restated]
For the asynchronous time model and for any graph $G$ of size $n$ with maximum degree $\Delta$, the stopping time of algebraic gossip is $O(\Delta n)$ rounds both in expectation and with high probability.
\end{theorem_main}

\begin{proof} 
Consider an arbitrary graph $G$ of size $n$ with a maximum degree $\Delta$ and a vertex $v$.
We pick any spanning tree rooted at $v$ and will only 
consider messages that are sent from the tree edges towards v, i.e., we 
ignore all messages received from non-tree edges or in the opposite 
direction (see Fig. \ref{fig:reduction_to_queues} (b)).

\begin{figure*}[ht]
\centering
\scalebox{1}{\begin{tikzpicture}
[inner sep=0.6mm, place/.style={circle,draw=black,fill=blue!20,thick,minimum size=4.5mm},>=stealth,place1/.style={circle,draw=black,fill=green!20,thick,minimum size=4.5mm}]

\def\myX{0.3}

\node at (0+\myX,0) [place] (a) {\tiny{$$}};
\node at (1.5+\myX,0) [place] (b) {\tiny{$$}};
\node at (0.7+\myX,-1) [place] (c)  {\tiny{$$}};
\node at (2.2+\myX,-1) [place1] (d) {\tiny{$v$}};
\node at (3+\myX,0) [place] (e)  {\tiny{$$}};
\node at (3.7+\myX,-1) [place] (f)  {\tiny{$$}};
\node at (2.2+\myX,-2) [place] (g)  {\tiny{$$}};
\node at (4.5+\myX,0) [place] (h)  {\tiny{$$}};

\path (a) edge (b);
\path (a) edge (c);
\path (b) edge (c);
\path (d) edge (f);
\path (h) edge (f);
\path (e) edge (c);
\path (d) edge (g);
\path (c) edge (g);
\path (b) edge (g);
\path (b) edge (e);
\path (h) edge (e);
\path (e) edge (f);

\node at (\myX+2.2, -2.8) [auto]{(a)};

\def\myX{8.3}

\node at (0+\myX,0) [place] (a) {\tiny{$$}};
\node at (1.5+\myX,0) [place] (b) {\tiny{$$}};
\node at (0.7+\myX,-1) [place] (c)  {\tiny{$$}};
\node at (2.2+\myX,-1) [place1] (d) {\tiny{$v$}};
\node at (3+\myX,0) [place] (e)  {\tiny{$$}};
\node at (3.7+\myX,-1) [place] (f)  {\tiny{$$}};
\node at (2.2+\myX,-2) [place] (g)  {\tiny{$$}};
\node at (4.5+\myX,0) [place] (h)  {\tiny{$$}};

\path[->,black,thick] (a) edge (b);
\path[opacity=0.09] (a) edge (c);
\path[opacity=0.09] (b) edge (c);
\path[<-,black,thick] (d) edge (f);
\path[->,black,thick] (h) edge (f);
\path[->,black,thick] (e) edge (f);
\path[opacity=0.09] (e) edge (c);
\path[<-,black,thick] (d) edge (g);
\path[->,black,thick] (c) edge (g);
\path[->,black,thick] (b) edge (g);
\path[opacity=0.09] (b) edge (e);
\path[opacity=0.09] (h) edge (e);

\node at (\myX+2.2, -2.8) [auto]{(b)};

\def\myX{0}
\def\myY{-4.3}

\foreach \x /\y/\thead/\ttail/\tname in {0/0/a/a2/v_1,2/0/b/b2/v_2,0/-1/c/c2/v_3,6/-1/v/v2/v,0/-2/e/e2/v_5,4/-2.5/f/f2/v_6,4/-0.5/g/g2/v_7,0/-3/h/h2/v_8}{
\node at (\x-0.95+\myX,\y+\myY) [auto] (\ttail)  {\tiny{}};
\shade[left color=yellow!20,right color=blue!10,draw=black] (\x-1+\myX,\y-0.18+\myY) rectangle (\x-0.24+\myX,\y+0.18+\myY);
\draw[white] (\x-1+\myX,\y+0.18+\myY) -- (\x-1+\myX,\y-0.18+\myY);

\shade[top color=green!80,draw=black!80] (\x-0.46+\myX,\y-0.18+\myY) rectangle (\x-0.24+\myX,\y+0.18+\myY);



%
%
%
%
%
%
\ifthenelse{\equal{\thead}{v}}{\node at (\x+\myX,\y+\myY) [place1] (\thead)  {\tiny{$\tname$}};}
{\node at (\x+\myX,\y+\myY) [place] (\thead)  {};}
}

\path[->,black,thick] (a) edge (b2);
\path[->,black,thick] (c) edge (g2);
\path[->,black,thick] (b) edge (g2);
\path[->,black,thick] (e) edge node [above,sloped,black] {{$\mu=p=\tfrac{1}{n\Delta}$}} (f2);
\path[->,black,thick] (h) edge (f2);
\path[->,black,thick] (g) edge (v2);
\path[->,black,thick] (f) edge (v2);

\node at (\myX+2.5, \myY-3.7) [auto]{(c)};

\def\myX{8}

\foreach \x /\y/\thead/\ttail/\tname in {0/0/a/a2/v_1,2/0/b/b2/v_2,6/0/v/v2/v,4/0/g/g2/v_7}{
\node at (\x-0.95+\myX,\y+\myY) [auto] (\ttail)  {\tiny{}};
\shade[left color=yellow!20,right color=blue!10,draw=black] (\x-1.25+\myX,\y-0.18+\myY) rectangle (\x-0.24+\myX,\y+0.18+\myY);
\draw[white] (\x-1.25+\myX,\y+0.18+\myY) -- (\x-1.25+\myX,\y-0.18+\myY);
\ifthenelse{\equal{\thead}{a}}
{
\shade[top color=green!80,draw=black!80] (\x-0.46+\myX,\y-0.18+\myY) rectangle (\x-0.24+\myX,\y+0.18+\myY);
\shade[top color=green!80,draw=black!80] (\x-0.46+\myX-0.22,\y-0.18+\myY) rectangle (\x-0.24+\myX-0.22,\y+0.18+\myY);
\shade[top color=green!80,draw=black!80] (\x-0.46+\myX-0.22-0.22,\y-0.18+\myY) rectangle (\x-0.24+\myX-0.22-0.22,\y+0.18+\myY);
\shade[top color=green!80,draw=black!80] (\x-0.46+\myX-0.22-0.22-0.22,\y-0.18+\myY) rectangle (\x-0.24+\myX-0.22-0.22-0.22,\y+0.18+\myY);
}

\ifthenelse{\equal{\thead}{b}}
{
\shade[top color=green!80,draw=black!80] (\x-0.46+\myX,\y-0.18+\myY) rectangle (\x-0.24+\myX,\y+0.18+\myY);
}

\ifthenelse{\equal{\thead}{v}}
{
\shade[top color=green!80,draw=black!80] (\x-0.46+\myX,\y-0.18+\myY) rectangle (\x-0.24+\myX,\y+0.18+\myY);

}

\ifthenelse{\equal{\thead}{g}}
{
\shade[top color=green!80,draw=black!80] (\x-0.46+\myX,\y-0.18+\myY) rectangle (\x-0.24+\myX,\y+0.18+\myY);
\shade[top color=green!80,draw=black!80] (\x-0.46+\myX-0.22,\y-0.18+\myY) rectangle (\x-0.24+\myX-0.22,\y+0.18+\myY);
}

\ifthenelse{\equal{\thead}{v}}{\node at (\x+\myX,\y+\myY) [place1] (\thead)  {\tiny{$\tname$}};}
{\node at (\x+\myX,\y+\myY) [place] (\thead)  {};}
}

\path[->,black,thick] (a) edge node [above,sloped,black, near start] {{$\mu$}}(b2);
\path[->,black,thick] (b) edge node [above,sloped,black, near start] {{$\mu$}}(g2);
\path[->,black,thick] (g) edge node [above,sloped,black, near start] {{$\mu$}}(v2);

\node at (\myX+2.5, \myY-1) [auto]{(d)};

\def\myY{-7}
\def\myXX{0.5}
\def\myYY{1.3}

\foreach \x /\y/\thead/\ttail/\tname in {0/0/a/a2/v_1,2/0/b/b2/v_2,6/0/v/v2/v,4/0/g/g2/v_7}{
\node at (\x-0.95+\myX,\y+\myY) [auto] (\ttail)  {\tiny{}};
\shade[left color=yellow!20,right color=blue!10,draw=black] (\x-1+\myX,\y-0.18+\myY) rectangle (\x-0.24+\myX,\y+0.18+\myY);
\draw[white] (\x-1+\myX,\y+0.18+\myY) -- (\x-1+\myX,\y-0.18+\myY);
\ifthenelse{\equal{\thead}{a}}
{
\shade[top color=green!80,draw=black!80] (\x-0.46+\myX+\myXX,\y-0.18+\myY+\myYY) rectangle (\x-0.24+\myX+\myXX,\y+0.18+\myY+\myYY);
\shade[top color=green!80,draw=black!80] (\x-0.46+\myX-0.22+\myXX,\y-0.18+\myY+\myYY) rectangle (\x-0.24+\myX-0.22+\myXX,\y+0.18+\myY+\myYY);
\shade[top color=green!80,draw=black!80] (\x-0.46+\myX-0.22-0.22+\myXX,\y-0.18+\myY+\myYY) rectangle (\x-0.24+\myX-0.22-0.22+\myXX,\y+0.18+\myY+\myYY);
\shade[top color=green!80,draw=black!80] (\x-0.46+\myX-0.22-0.22*6+\myXX,\y-0.18+\myY+\myYY) rectangle (\x-0.24+\myX-0.22-0.22*6+\myXX,\y+0.18+\myY+\myYY);
\shade[top color=green!80,draw=black!80] (\x-0.46+\myX-0.22-0.22*2+\myXX,\y-0.18+\myY+\myYY) rectangle (\x-0.24+\myX-0.22-0.22*2+\myXX,\y+0.18+\myY+\myYY);
\shade[top color=green!80,draw=black!80] (\x-0.46+\myX-0.22-0.22*3+\myXX,\y-0.18+\myY+\myYY) rectangle (\x-0.24+\myX-0.22-0.22*3+\myXX,\y+0.18+\myY+\myYY);
\shade[top color=green!80,draw=black!80] (\x-0.46+\myX-0.22-0.22*4+\myXX,\y-0.18+\myY+\myYY) rectangle (\x-0.24+\myX-0.22-0.22*4+\myXX,\y+0.18+\myY+\myYY);
\shade[top color=green!80,draw=black!80] (\x-0.46+\myX-0.22-0.22*5+\myXX,\y-0.18+\myY+\myYY) rectangle (\x-0.24+\myX-0.22-0.22*5+\myXX,\y+0.18+\myY+\myYY);
}

\ifthenelse{\equal{\thead}{v}}{\node at (\x+\myX,\y+\myY) [place1] (\thead)  {\tiny{$\tname$}};}
{\node at (\x+\myX,\y+\myY) [place] (\thead)  {};}
}

\path[->,black,thick] (a) edge node [above,sloped,black, near start] {{$\mu$}}(b2);
\path[->,black,thick] (b) edge node [above,sloped,black, near start] {{$\mu$}}(g2);
\path[->,black,thick] (g) edge node [above,sloped,black, near start] {{$\mu$}}(v2);

\draw[->,thick] (-0.46+\myX+\myXX+0.22, \myY+\myYY) .. controls(-0.46+\myX+\myXX+1+0.54, \myY+\myYY-1.3) and (-0.46+\myX+\myXX+0.22-2.5, \myY+\myYY-0.3) .. node [above,sloped,black] {{\small{$\lambda=\mu/2$}}}(-1+\myX, \myY);

\node at (\myX+2.5, \myY-1) [auto]{(e)};

\end{tikzpicture}}
\caption{\small{Reduction of Algebraic Gossip to a system of queues. (a) -- Initial graph $G$. (b) -- Spanning tree rooted at $v$, $G_v$. (c) -- System of queues $Q_{n}^{tree}$. (d) -- System of queues $Q_{l_{\max}}^{line}$. Stopping time of $Q_{l_{\max}}^{line}$ is larger than of $Q_{n}^{tree}$. (e)--Taking all customers out of the system and use Jackson theorem for open networks.}}
\label{fig:reduction_to_queues}
\end{figure*}
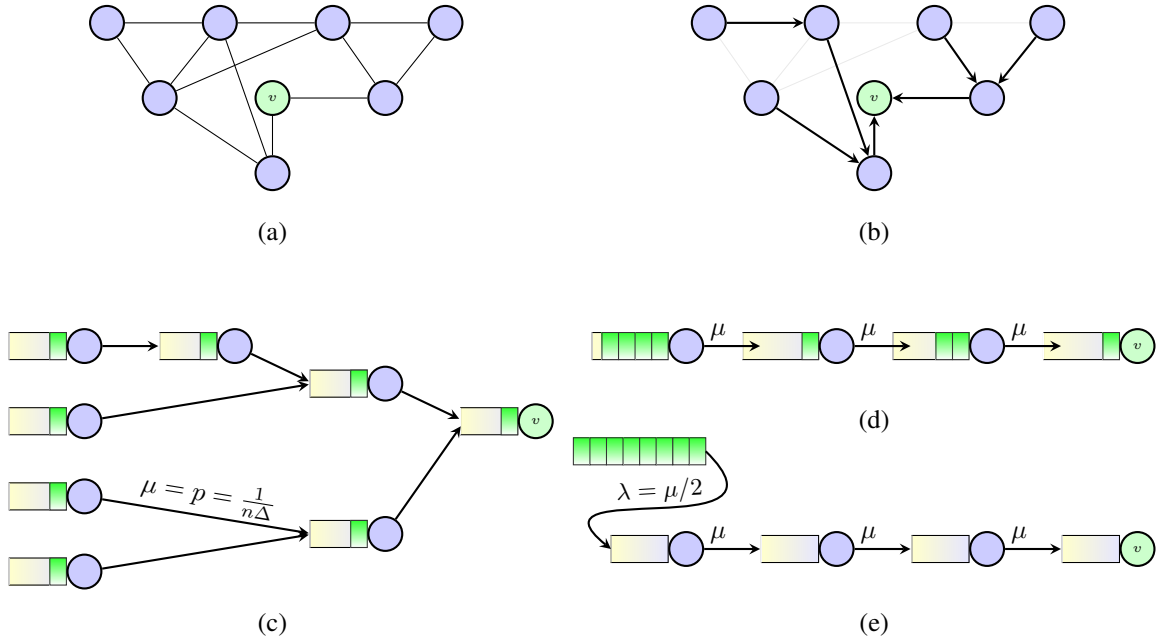

Now, let us concentrate on the information flow towards node $v$ from all other nodes. As in the proof of Theorem \ref{thm:ring_is_linear}, we will define a queuing system with a queue at each node (see Fig. \ref{fig:reduction_to_queues} (c)). The following lemma shows that we can model the service time at each queue as an exponential random variable with parameter $\mu=p$. 

Let $G_v$ be a tree of size $n$ rooted by  node $v$. Let $\mathcal{N}(G_v, \mathcal{S})$ be a network of $n$ queues where for each node $u$ in $G_v$ there is a queue and the queue output is connected to the input of the queue corresponding to the parent of $u$ in $G_v$. In addition, each queue is of infinite size and initially has one customer in the queue (see Fig. \ref{fig:reduction_to_queues}) (c)). The servers of all the queues work with a service time distributed as $\mathcal{S}$.
Let $T(G_v, \mathcal{S})$ be the random time by which all the $n$ customers in $\mathcal{N}(G_v, \mathcal{S})$ arrive to the queue of $v$ (we assume $v$ does not serve the customers).

\begin{lemma}\label{lemma:exp_server_instead_of_geom}
For any tree $G_v$ and $0 < p \le 1$:
\begin{align*}
&\Pr\left(T(G_v, \text{Geom}(p))\leq t\right)\geq\Pr\left(T(G_v, \text{Exp}(p))\leq t\right),\\
&\text{ for all } t\ge0.
\end{align*}
\end{lemma}

The result of this lemma is that any probabilistic upper bound on the stopping time of $v$ in a tree network with exponential servers holds for the same tree network with geometric servers (both with the same parameter $p$ and initially one customer at each queue).

Once all \emph{real} customers arrive at $v$, it will reach rank $n$ and will finish the algebraic gossip task. Now we have to calculate the service time parameter $p$. The degree of each node in $G$ is at most $\Delta$. 
Each node in $G_v$, except $v$, has a parent.
Since we virtually remove (i.e., ignore) all edges that do not belong to $G_v$, at each node there is exactly one edge that goes towards the root $v$. Therefore, the probability that a customer will be serviced (transmitted towards $v$) at the end of a given timeslot is at least:
$p\geq \left(\tfrac{2}{n}\cdot\tfrac{1}{\Delta}\right)(1-\tfrac{1}{q})$,
where $\frac{2}{n}\cdot \tfrac{1}{\Delta}=\tfrac{2}{n\Delta}$ is the probability that in the \ex algorithm a message will be sent on the edge that goes towards $v$ during one timeslot, and $(1-\frac{1}{q})$ is the minimal probability that the message is \emph{helpful} (Lemma \ref{lemma_helpful_rank}).
Clearly, $p\geq\frac{1}{n\Delta}$ for $q\geq 2$, so we set our exponential servers to work with rate $\mu=\frac{1}{n\Delta}$.

\begin{theorem}
\label{thm:tree_of_queues}
Let $Q_n^{tree}$ be a network of $n$ nodes arranged in a tree topology, rooted at the node $v$.
Each node has an infinite queue, and a single exponential server with parameter $\mu$. Initially, there is a single customer in every queue. The time by which all $n$ customers leave the network via the root node $v$ is $t({Q}_n^{tree})=O(n/\mu)$ with high probability. Formally, for any $\alpha>1$:
\begin{align}
\Pr \left(t({Q}_n^{tree})< \alpha 4n/\mu\right) > 1-2(2e^{-\alpha/2})^n.
\end{align}
\end{theorem}

The main idea of the Theorem \ref{thm:tree_of_queues} proof is to show that the stopping time of the network $Q_n^{tree}$ (i.e., the time by which all the customers leave the network) is stochastically \footnote{For completeness, stochastic dominance is formally defined in Appendix \ref{app:K-1}.} smaller or equal to the stopping time of the systems of $l_{\max}$ queues arranged in a line topology -- $Q_{l_{\max}}^{line}$ ($l_{\max}$ is the depth of the tree $Q_n^{tree}$). Then, we make the system $Q_{l_{\max}}^{line}$ stochastically slower by moving all the customers out and make them enter the system via the farthest queue with the rate $\lambda=\mu/2$. Finally, we use Jackson's Theorem for open networks
(similar to the proof of Theorem \ref{thm:ring_is_linear}) 
to find the stopping time of the system. See Fig. \ref{fig:reduction_to_queues} for the illustration. The full proof of Theorem \ref{thm:tree_of_queues} can be found in Appendix \ref{app:K}.

Using Theorem \ref{thm:tree_of_queues} for the tree $G_v$ and with $\mu =\tfrac{1}{\Delta n}$, we obtain the stopping time of the node $v$: $T_v < \alpha 4n^2\Delta$ with probability of at least  $1-2(2e^{-\alpha/2})^n$. 

The same analysis holds for any node $u\in V$, i.e., we consider a spanning tree rooted at $u$, define a queuing system on the tree $G_u$, and find the stopping time of $u$, $T_u$. So, we can use a union bound to obtain the stopping time of all the nodes in $G$:
\begin{align}\label{eq:t_upper}
\Pr \left(\bigcap_{u\in V} (T_u <\alpha 4n^2\Delta)\right) \ge 1-2n(2e^{-\alpha/2})^n.
\end{align}
%
%
%
%
By letting $\alpha = 2$ we obtain:
\begin{align}
\Pr \left(\bigcap_{u\in V} (T_u <8n^2\Delta)\right) \ge 1-2n(\tfrac{2}{e})^n.
\end{align}
So, we determined that the stopping time of the algebraic gossip in $G$ is $O(\Delta n^2)$ timeslots with high probability and thus: 
$\hat{R}=O(\Delta n).$

The high probability bound of Eq. \eqref{eq:t_upper} is true for any $\alpha > 1$ and therefore strong enough to bound the expectation (see proof in the Appendix \ref{app:F}) and finish the proof of Theorem \ref{thm:upper_bound_any_graph}:
\begin{align}\label{eq:expected}
E[T] = O(\Delta n^2)\text{ and } E[R]=O(\Delta n).
\end{align}
\end{proof}

From Theorem \ref{thm:upper_bound_any_graph}, and since the maximum degree is at most $n$ we can derive a general upper bound of algebraic gossip on any graph.
\begin{corollary}\label{cor:upper}
For the asynchronous time model and any graph $G$ of size $n$, the gossip stopping time of the algebraic gossip task is $O(n^2)$ rounds, both in expectation and with high probability.
\end{corollary}

We can use Theorem \ref{thm:upper_bound_any_graph} to obtain a tight linear bound of algebraic gossip on graphs with a constant maximum degree. We note that previous bounds for this case 
are not tight, for example, for the ring graph the bound of \cite{Mosk-Aoyama2006Information} is $O(n^2)$.

\begin{corollary}
\label{corollary_any_ex_nc_linear}
For the asynchronous time model and any graph $G$ of size $n$ with a constant maximum degree $\Delta $, the gossip stopping time of the algebraic gossip task is $O(n)$ rounds both in expectation and with high probability.
\end{corollary}

We now show that the upper bound $O(\Delta n)$, presented in Theorem \ref{thm:upper_bound_any_graph}, is tight in the sense that 
for almost any $\Delta$
there exists a graph for which algebraic gossip takes $\Omega(\Delta n)$ rounds.

\begin{theorem}
\label{thm:lower_bound_any_graph}
For any constant $\epsilon>0$ and $2 \le \Delta \le (1-\epsilon)n$, and for the asynchronous time model there exists a graph $G$ of size $n$ with maximum degree $\Delta$ for which algebraic gossip takes $\Omega(\Delta n)$ rounds both in expectation and with high probability. In particular, there is a graph for which the stopping time is $\Omega(n^2)$ rounds both in expectation and with high probability.
\end{theorem}

\begin{proof}
In order to prove this result we will need the following lemma; the proof can be found in Appendix \ref{app:G}:

\begin{lemma}
\label{lemma:geom_lower_bound}
Let $X$ be a sum of $m$ independent and identically distributed geometric random variables with parameter $p$, i.e., $X=\sum_{i=1}^{m} X_i$. Then, for any positive integer $k<m/p$
\begin{align}
\Pr \left(X > k\right) \geq 1-\left(\frac{m}{e^{\frac{m-kp}{m}}kp}\right)^{-m}.
\end{align}
\end{lemma}

Let us construct a graph $G$ with $\left|V(G)\right|=n$ nodes and maximum degree $\Delta(G)=\Delta$. 
Consider two arbitrary graphs $G'$ and $G''$ with certain maximum degrees $\Delta(G')$ and $\Delta(G'')$, respectively, and with total number of nodes $n$ ($\left|V(G')\right|+\left|V(G'')\right|=n$).
We now distinguish two cases: $\Delta\le n/2$ and $\Delta>n/2$.
For the first case ($\Delta\le n/2$), let $u\in V(G')$ and $v\in V(G'')$, such that $d_u=\Delta(G')=\Delta-1$ and $d_v=\Delta(G'')=\Delta-1$. 
We construct $G$ by interconnecting $G'$ and $G''$ with a new edge $(u,v)$, i.e., $V(G)=V(G')\cup V(G'')$ and $E(G)=E(G')\cup E(G'')\cup (u,v)$. See Fig. \ref{fig:fig/barbells.pdf} (a) for illustration.

For the second case ($\Delta> n/2$), the only difference in construction of $G$ is the degree of $v\in V(G'')$, which is now $d_v=\Delta(G'')=n-\Delta-1$.

In order to finish algebraic gossip on $G$, at least $\max \left\{|V(G')|,|V(G'')|\right\}\ge\frac{n}{2}$ messages should be sent over the edge $(u,v)$.
Using the fastest gossip variation -- \ex, the probability $p$ that a \emph{helpful message} will be sent in one timeslot over the edge $(u,v)$ is bounded by the probability that any message will be sent over $(u,v)$, so: $p\le \tfrac{1}{n}\left(\tfrac{1}{\Delta(G')+1}+\tfrac{1}{\Delta(G'')+1}\right)$.

For the first case ($\Delta\le n/2$) we obtain: 
\begin{align}
p\le \tfrac{1}{n}\left(\tfrac{1}{\Delta}+\tfrac{1}{\Delta}\right)=\tfrac{2}{n\Delta}.
\end{align}
For the second case we get:
\begin{align}
p&\le \tfrac{1}{n}\left(\tfrac{1}{\Delta}+\tfrac{1}{n-\Delta}\right)=\tfrac{1}{\Delta(n-\Delta)}\\
&\le\tfrac{1}{\Delta(n-(1-\epsilon)n)}=\tfrac{1}{n\Delta\epsilon}.
\end{align}
%
We can see that the first case can be viewed as the second with $\epsilon=0.5$; thus, we can further analyze only the second case.
The number of timeslots, $T$, needed to send $n/2$ \emph{helpful messages} over the edge $(u,v)$, can be viewed as a sum of $n/2$ geometric random variables with parameter $p$. Clearly, $\text{E}\left[T\right]=\tfrac{n}{2}\cdot\tfrac{1}{p}=\tfrac{n^2\Delta\epsilon}{2}=\Omega(\Delta n^2)$ timeslots in both cases.
Using Lemma \ref{lemma:geom_lower_bound} with $k=\left\lfloor \text{E}\left[T\right]/2\right\rfloor=\left\lfloor n^2\Delta\epsilon/4\right\rfloor$, $p=\frac{1}{n\Delta\epsilon}\cdot\frac{n^2\Delta\epsilon/4}{\left\lfloor n^2\Delta\epsilon/4\right\rfloor}\ge\frac{1}{n\Delta\epsilon}$ (we took $p$ even larger than its maximum value; this will make calculations nicer and will not affect the bound), and $m=n/2$ we get:
\begin{align}
\Pr \left(T > \left\lfloor n^2\Delta\epsilon/4\right\rfloor\right) &\geq 1-\left(\frac{m}{e^{\frac{m-kp}{m}}kp}\right)^{-m} 
\\&= 1-\left(\sqrt{e}/2\right)^{n/2}.
\end{align}

It is clear that $\Pr \left(T\geq k\right)$ increases when $p$ decreases (the smaller probability of success -- the larger the probability to finish later). Hence, the above inequality holds for any $p\leq \frac{1}{n\Delta\epsilon}$.

Thus, the number of timeslots needed is at least $\left\lfloor n^2\Delta\epsilon/4\right\rfloor$ w.h.p. and $n^2\Delta\epsilon/2$ in expectation. So, the total stopping time of the algebraic gossip protocol on the graph $G$ (measured in rounds) is: $\hat{R}=\Omega(\Delta n)$, and $\text{E}\left[R\right]=\Omega(\Delta n)$, where $2\le\Delta\le(1-\epsilon)n$ for any constant $\epsilon>0$. The lower bound of $\Omega(n^2)$ rounds is achieved, for example, in a barbell graph -- two cliques interconnected with a single edge (see Fig. \ref{fig:fig/barbells.pdf} (b)).
\end{proof}

%
%
%

\begin{figure}
\centering
\includegraphics[width=3.4in]{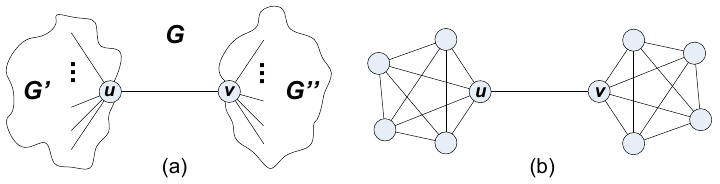}
\caption{(a) Graph $G$, constructed from $G'$ and $G''$, for the proof of Theorem \ref{thm:lower_bound_any_graph}. (b) An example of a $G$ graph with $\Delta(G)=n/2$: barbell graph (two cliques of size $n/2$ connected with a single edge).\label{fig:fig/barbells.pdf}}
\end{figure}

\section{\ex Can Be Unboundedly Faster Than \push or \pull}
\label{sec:ex_is_faster}
As we presented earlier, there are three gossip variations: \push, \pull, and \ex. In \push or \pull there is only one message sent between the communication partners, in \ex two messages are sent. Thus, the total message complexity for the same number of communication rounds is doubled. We would like to know: Is the stopping time decrease when using \ex worth the doubling message complexity?
In this section we give the answer by presenting a graph for which the \ex gossip algorithm is \emph{unboundedly} faster than the \push or \pull algorithms.

\begin{theorem}\label{theorem_ex_is_better_than_push_pull}
For the star graph $S_n$ (which is a tree of $n$ nodes with one node having degree $n-1$ and the other $n-1$ nodes having degree 1), algebraic gossip using \ex is unboundedly better than using \push or \pull algorithms. Formally, for $\A \in \{\push, \pull\}$:
\begin{align}
&\lim_{n\rightarrow\infty}\frac{\hat{R}(\A)}{\hat{R}(\ex)}\rightarrow\infty \text{ , and }
\\&\lim_{n\rightarrow\infty}\frac{E\left[R(\A)\right]}{E\left[R(\ex)\right]}\rightarrow\infty.
\end{align}

\end{theorem}
The proof of this theorem is a direct consequence of the following lemmas. 
\begin{lemma}\label{lemma:push_in_star_lower_bound}
For the star graph $S_n$, 
algebraic gossip using \push takes $\Omega(n^2)$ rounds 
with high probability and in expectation.
\end{lemma}

\begin{proof}
We are interested in a lower bound, so we will consider the minimum number of rounds to complete the task. The center node can finish the algorithm after one round since all other nodes will send (\push) to it their messages and in the best case all these messages will be \emph{helpful}, so we ignore this phase.
Now, the center node should send (\push) to every other node $n-1$ independent linear equations. In the synchronous time model, the center node wakes up exactly once in a round. Thus, the number of rounds needed to \push $n-1$ messages to all the $n-1$ other nodes is at least $(n-1)\cdot(n-1)$ with probability $1$. In the asynchronous model, the center node will wake up in a given timeslot with probability $1/n$; thus, it will need $\Omega(n\cdot(n-1)\cdot(n-1))$ timeslots (to \push $n-1$ messages to all the $n-1$ other nodes) in expectation and with the high probability (sum of $n$ independent geometric random variables). Thus, for both time models, the number of rounds needed is $\Omega(n^2)$.
\end{proof}

\begin{lemma}\label{lemma:pull_in_star_lower_bound}
For the star graph $S_n$, 
algebraic gossip using \pull takes $\Omega(n\log n)$ rounds with 
high probability and in expectation.
\end{lemma}

\begin{proof}
First, we give the following claim for the \emph{coupon collector problem} \cite{1076315}. The proof of the claim can be found in Appendix \ref{app:H}.
\begin{clm}\label{clm-coupons}
Let $X$ be the r.v. for the number of coupons needed to obtain $n$ distinct coupons (i.e., to obtain at least one coupon of each type), then:
$$E[X] = \Theta(n\log n)
\: \: \text{ and w.h.p. } \: \:
X = \Theta(n\log n).$$
\end{clm}
The center node will finish the algorithm once it receives (\pull) a helpful message from every other node. Thus, the center node has to reach (\pull) every other node at least once. In the synchronous time model, the center node will transmit (wake up) exactly once in a round. Reaching every other node at least one time is exactly the coupon collector problem, so (using Claim \ref{clm-coupons}): $\hat{R}=\Omega(n\log n)$, and $\text{E}\left[R\right]=\Omega(n\log n)$ rounds. In the asynchronous model, the center node will wake up in a given timeslot with probability $1/n$; thus, it needs $\Omega(n\cdot n\log n)$ timeslots in order to wake up $\Omega(n\log n)$ times in expectation and with high probability (lower bound on sum of i.i.d. geometric r.v.'s).
\end{proof}

\begin{lemma}
\label{lemma_star_nc_ex}
For the star graph $S_n$, 
algebraic gossip using \ex takes $O(n)$ rounds with high probability and in expectation.
\end{lemma}
\begin{proof}
To prove Lemma \ref{lemma_star_nc_ex} we will use the following claim. The proof can be found in Appendix \ref{app:I}.

\begin{clm}
\label{lemma_sum_of_geom_vars}
Let $X_i$ be independent geometric random variables with parameter $p$, and let $X=\sum_{i=1}^n X_i$.
For $p\geq \frac{1}{2}$, and $\alpha>1$:
\begin{align}
\Pr \left(X\geq 2n\alpha\right)\leq \left(2^{1.5-\alpha}\right)^n.
\end{align}
\end{clm}

First, we consider the synchronous time model. Let us split the task into two phases. The first phase is the time (in rounds) $R_1$ until the center node $v_1$ learns all the initial messages, i.e., $dim(S_{v_1}(t))=n$. The second phase is the time (in rounds) $R_2$ it takes $v_1$ to distribute the information to all the nodes.

Initially, every node $u \in V \setminus \{v_1\}$ is \emph{helpful} to $v_1$. By Lemma \ref{lemma_helpful_rank}, a message sent from $u$ to $v_1$ will be \emph{helpful} with probability of at least $1-\tfrac{1}{q}$; thus, after $n$ rounds, a node $u$ will send a \emph{helpful message} to $v_1$ with probability of at least $1-\left(\frac{1}{2}\right)^n$ (for $q>2$). Using union bound we can find the probability that all the nodes $u \in V \setminus \{v_1\}$ will send a helpful message to $v_1$ after $n$ rounds:
\begin{align}
\Pr(R_1>n)\le \sum_{u \in V \setminus \{v_1\}}\left(\frac{1}{2}\right)^n\le n\left(\frac{1}{2}\right)^n.
\end{align}

Now, from the beginning of phase two, $dim(S_{v_1})=n$ and hence the node $v_1$ will be \emph{helpful} to every other node until the rank of that node becomes $n$. From Lemma \ref{lemma_helpful_rank}, a message transmitted to some node from a node \emph{helpful} to it, will increase its dimension with probability $p \geq 1-\frac{1}{q}$.

Let us define $X_i^u$ as the number of rounds needed for $v_1$ to increase the rank of some node $u\in{V}\backslash{\left\{v_1\right\}}$. It is clear that $X_i^u$ has a geometric distribution with parameter $p$. We are interested to find $X^u=\sum_{i=1}^n X_i^u$, which represents the number of rounds by which the rank of node $u$ will become $n$. Using Claim \ref{lemma_sum_of_geom_vars} (and the fact that for $q>2$, $p=1-\frac{1}{q}>\frac{1}{2}$), we obtain for any $\alpha>1$ that:
\begin{align}
\Pr(X^u<2\alpha n)\ge 1-(2^{1.5-\alpha})^n.
\end{align}
Using union bound, we obtain the probability that ranks of all nodes will become $n$ after $2\alpha n$ rounds:
\begin{align}
\Pr \left(\cup_{u\in{V}\backslash{\left\{v_1\right\}}} X^u \ge 2\alpha n\right)&\le \sum_{u\in{V}\backslash{\left\{v_1\right\}}}\Pr(X^u\ge2\alpha n)
\\ &\le n\left(2^{1.5-\alpha}\right)^n,
\end{align}
and thus:
\begin{align}
\Pr \left(\cap_{u\in{V}\backslash{\left\{v_1\right\}}} X^u<2\alpha n\right)\geq 1-n\left(2^{1.5-\alpha}\right)^n.
\end{align}

So, 
\begin{align}\label{s_whp}
\Pr(R_2<2\alpha n)\ge 1-n\left(2^{1.5-\alpha}\right)^n,
\end{align}
and for $\alpha=2$: 
\begin{align}
\Pr(R_2<4n)\ge 1-n\left(\frac{1}{\sqrt{2}}\right)^n.
\end{align}

Combining the two phases together, i.e., $R\le R_1+R_2$, we have:
\begin{align}
\Pr(R>5n)&\le \Pr(R_1\ge n) + \Pr(R_2\ge 4n)
\\&\le n\left(\frac{1}{2}\right)^n + n\left(\frac{1}{\sqrt{2}}\right)^n
\\&\le 2n\left(\frac{1}{\sqrt{2}}\right)^n,
\end{align}
and thus: $\hat{R}=O(n)$.

Let us now find an upper bound for the expected number of rounds needed to complete the task, $E[R]$.
Since $R\le R_1+R_2$, we get: $\text{E}\left[R\right]\le \text{E}\left[R_1\right]+\text{E}\left[R_2\right]$. 
During the first phase, each node $u \in V \setminus \{v_1\}$ will send a \emph{helpful message} to $v_1$ with probability of at least $\tfrac{1}{2}$. Thus, $\text{E}\left[R_1\right]\le 2n$. The high probability bound of \eqref{s_whp} allow us to show that for sufficient large $n$:
\begin{align}\label{s_expected}
E[R_2] \le 4n + 1
\end{align}
The proof of Eq. \eqref{s_expected} can be found in Appendix \ref{app:J}.
Hence, we obtain: $\text{E}\left[R\right]=O(n)$. In order to justify the result for the asynchronous time model (in which a node wakes up at a given timeslot with probability $1/n$), we notice that a node $v_1$ will wake up $O(n)$ times after at most $O(n^2)$ timeslots (or $O(n)$ rounds) with expectation and with high probability (sum of $n$ i.i.d. geometric r.v.'s). Thus, the lemma holds for both time models.
\end{proof}

Since for $\A\in\{\push,\pull\}$: $\hat{R}(\A)=\Omega(n\log n)$ and $\text{E}\left[R(\A)\right]=\Omega(n\log n)$ (Lemmas \ref{lemma:push_in_star_lower_bound} and \ref{lemma:pull_in_star_lower_bound}), and from Lemma \ref{lemma_star_nc_ex}: $\hat{R}(\ex)=O(n)$ and $\text{E}\left[R(\ex)\right]=O(n)$, we are ready to conclude that:
\begin{align*}
&\lim_{n\rightarrow\infty}\frac{\hat{R}(\A)}{\hat{R}(\ex)}\rightarrow\infty \text{ , and }
\\&\lim_{n\rightarrow\infty}\frac{E\left[R(\A)\right]}{E\left[R(\ex)\right]}\rightarrow\infty.
\end{align*}

\section{Conclusions}\label{sec:conclusions}

In this work we prove bounds on the stopping time of the algebraic gossip protocol. We prove that the upper bound for any graph is $O(n^2)$ and we show that this bound is tight in a sense that there exists a graph for which the stopping time of algebraic gossip is $\Omega(n^2)$. Our general upper bound $O(\Delta n)$ is provided as a function of the maximum degree  $\Delta$ of a graph and thus we can obtain a tight linear bound of $\Theta(n)$ for any graph with a constant maximum degree. Moreover, our results hold for $q \ge 2$ (coefficients field size), while previous results were for the case $q \ge n$.

It is still an open question to fully understand the properties of a network that 
captures the stopping time of algebraic gossip. To illustrate this, note the 
interesting observation that on the extended-barbell graph (Fig. \ref{fig:fig/ext_barbell.pdf}) the stopping time of algebraic gossip is linear. So, by adding a single node to the barbell graph (Fig. \ref{fig:fig/barbells.pdf} (b)) the  stopping time has been changed by an order of magnitude?!

\begin{figure}
\centering
\includegraphics[width=2in]{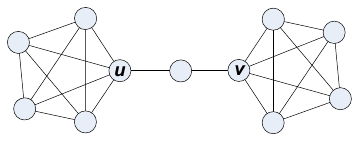}
\caption{Extended barbell graph: additional node between the cliques.\label{fig:fig/ext_barbell.pdf}}
\end{figure} 

In the conference version of this paper \cite{Borokhovich2010Tight}, we originally asked the above question, and a recent work by Haeupler \cite{Haeupler2010Analyzing} makes significant progress in answering it. 
While all previous works on algebraic gossip used the notion of \emph{helpful message/node} to look at the rank evaluation of the matrices each node maintains (this approach was initially proposed by \cite{Deb2006Algebraic}), Haeupler used a completely different approach. Instead of looking on the growth of the node's subspace (spanned by the linear equations it has), he proposed to look at the orthogonal complement of the subspace and then analyze the process of its disappearing. This elegant and powerful approach led to a very impressive result: a tight bound of $\Theta(n/\gamma)$ was proposed for all-to-all communication, where $\gamma$ is a min-cut measure of the related graph. This matches our tight bounds for many topologies (e.g., constant maximum degree graphs, barbell graph, etc.) and extends them for other graphs. Haeupler also proposed results for many-to-all ($k$ to $n$) communication, but these bounds are not always tight. 

In our recent conference paper \cite{Avin2011OrderOptimal}, we successfully addressed the topics of many-to-all communication and the non-uniform gossip approach. These topics were originally raised in the conference version \cite{Borokhovich2010Tight} of the current paper. First, we provide in \cite{Avin2011OrderOptimal} an upper bound for the many-to-all scenario and show that the bound is tight for various topologies (in particular, for graphs with a constant maximum degree); second, we study a non-uniform gossip and propose a modified algebraic gossip algorithm that is order optimal for many families of graphs. This recent work \cite{Avin2011OrderOptimal} is based on the queuing analysis technique that is novelly proposed in the current manuscript.

\renewcommand{\abstractname}{Acknowledgments}

\begin{abstract}
This research was supported in part by GIF - Grant No. 2183-1807.6/2007.
We would like to thank the anonymous reviewers who helped us to significantly improve the paper.
\end{abstract}

\bibliographystyle{acm}



\appendix


\subsection{Algebraic gossip with synchronous time model, with \push and \pull}\label{app:A}
In this section of Appendix we give theorems and corollaries that extend the results presented in the paper to both time models (synchronous and asynchronous), and to the three gossip algorithms (\push, \pull, and \ex).

The first theorem shows that the general upper bound for algebraic gossip also holds for the synchronous time model.
\begin{theorem}\label{thm:general_bound_synch}
For the \textbf{synchronous} time model and for any graph $G_n$ with maximum degree $\Delta$, the stopping time of algebraic gossip is $O(\Delta n)$ rounds with high probability.
\end{theorem}

\begin{proof}
The proof for the synchronous time model is almost the same as in the asynchronous case. The analysis will be done in \emph{rounds} instead of \emph{timeslots}.
The probability that a customer will be serviced (transmitted towards $v$) at the end of a given \emph{round} is at least:
$p\geq \left(1-(\tfrac{\Delta-1}{\Delta})^2\right)(1-\tfrac{1}{q})$,
where $\left(1-(\tfrac{\Delta-1}{\Delta})^2\right)=\tfrac{2}{\Delta}-\tfrac{1}{\Delta^2}$ is the probability that in the \ex algorithm at least one message will be sent on a specific edge (in a specific direction) during one \emph{round}, and $(1-\frac{1}{q})$ is the minimal probability that the message is \emph{helpful} (Lemma \ref{lemma_helpful_rank}).
$$p\geq (\tfrac{2}{\Delta}-\tfrac{1}{\Delta^2})(1-\tfrac{1}{q})\geq  \tfrac{1}{\Delta}(1-\tfrac{1}{q})\geq\tfrac{1}{2\Delta} \text{ for }q\geq 2.$$

If node $i$ has received a message during a specific round from node $j$ it will ignore the additional message that can arrive from the same node $j$ at the same round (this can happen if $i$ chooses $j$ and $j$ chooses $i$ in the \ex gossip scheme in one round). Clearly, this assumption can only increase the stopping time since we ignore (possibly helpful) information.

$T^x, \text{ }x\in{ \left\{\overrightarrow{\text{arr}},\overrightarrow{\text{cross}}\right\}}$ are measured now in \emph{rounds} and not in \emph{timeslots}.
Since $\mu=p\geq\frac{1}{2\Delta}$, using Lemma \ref{lemma:sum_of_exp_bounded1} (with $\alpha=2$ and $\text{E}\left[T^x\right]=\tfrac{2n}{\mu}=4n\Delta$), we obtain: 
$$\Pr \left(T^x < 8n\Delta\right) > 1-\left(\tfrac{2}{e}\right)^n\text{ , for }x\in{ \left\{\overrightarrow{\text{arr}},\overrightarrow{\text{cross}}\right\}}.$$

The rest of the proof is the same as in the asynchronous case and thus the result follows.
\end{proof}

The following theorem proves that the worst-case lower bound for algebraic gossip also holds for the synchronous time model.
\begin{theorem}\label{thm:lower_bound_any_graph_sync}
For any constant $\epsilon>0$ and $2 \le \Delta \le (1-\epsilon) n$, and for the synchronous time model, there exists a graph $G$ of size $n$ with maximum degree $\Delta$ for which algebraic gossip takes $\Omega(\Delta n)$ rounds both in expectation and with high probability. In particular, there is a graph for which the stopping time is $\Omega(n^2)$ rounds both in expectation and with high probability.
\end{theorem}

\begin{proof}
The proof is almost the same as in Theorem \ref{thm:lower_bound_any_graph}. The analysis will be done in \emph{rounds} instead of \emph{timeslots}.

Using the fastest gossip variation -- \ex, the probability $p$ that a \emph{helpful message} will be sent in one timeslot over the edge $(u,v)$ can be bounded (using a union bound) as: $p\le \tfrac{1}{\Delta(G')+1}+\tfrac{1}{\Delta(G'')+1}$.

For the first case ($\Delta\le n/2$) we obtain: $p\le\tfrac{2\Delta}{\Delta\cdot\Delta}=\tfrac{2}{\Delta}$.
For the second case ($\Delta > n/2$) we get:
\begin{align}
p&\le \tfrac{1}{\Delta}+\tfrac{1}{n-\Delta}=\tfrac{n}{\Delta(n-\Delta)}\\
&\le\tfrac{n}{\Delta(n-(1-\epsilon)\Delta)}=\tfrac{1}{\Delta\epsilon}.
\end{align}

The first case can be viewed as the second with $\epsilon=0.5$; thus, we can further analyze only the second case.

%
%
%
The number of rounds, $R$, needed to to send $n/2$ \emph{helpful} messages over the edge $(v,u)$, can be viewed as a sum of $n/2$ geometric random variables with parameter $p$. Clearly, $\text{E}\left[R\right]=\tfrac{n}{2}\cdot\tfrac{1}{p}=\tfrac{n\Delta(1-\alpha)}{4}=O(\Delta n)$ rounds.
The number of rounds, $R$, needed to to send $n/2$ \emph{helpful messages} over the edge $(u,v)$, can be viewed as a sum of $n/2$ geometric random variables with parameter $p$. Clearly, $\text{E}\left[R\right]=\tfrac{n}{2}\cdot\tfrac{1}{p}=\tfrac{n\Delta\epsilon}{2}=O(\Delta n)$ rounds in both cases.
Using Lemma \ref{lemma:geom_lower_bound} with $k=\left\lfloor \text{E}\left[R\right]/2\right\rfloor=\left\lfloor n\Delta\epsilon/4\right\rfloor$, $p=\frac{1}{\Delta\epsilon}\cdot\frac{n\Delta\epsilon/4}{\left\lfloor n\Delta\epsilon/4\right\rfloor}\ge\frac{1}{\Delta\epsilon}$ (we took $p$ even larger than its maximum value; this will make calculations nicer and will not affect the bound), and $m=n/2$ we get:
\begin{align}
\Pr \left(R > \left\lfloor n\Delta\epsilon/4\right\rfloor\right) &\geq 1-\left(\frac{m}{e^{\frac{m-kp}{m}}kp}\right)^{-m} 
\\&= 1-\left(\sqrt{e}/2\right)^{n/2}.
\end{align}The rest of the proof is the same as in Theorem \ref{thm:lower_bound_any_graph}.
\end{proof}

The following corollary shows that our bounds (upper and lower) for algebraic gossip on general graphs hold also for the \push and \pull gossip algorithms.
\begin{corollary}\label{cor:push_pull}
The results of Theorems \ref{thm:upper_bound_any_graph}, \ref{thm:lower_bound_any_graph}, \ref{thm:general_bound_synch}, and \ref{thm:lower_bound_any_graph_sync} hold also for \push and \pull gossip algorithms.
\end{corollary}

\begin{proof}
By moving from the \ex to the \push or \pull gossip algorithms, we change only the probability of sending a \emph{helpful message} on a specific (directed) edge, i.e., the service time at each node will change. Easy to see that this probability will be decreased by a factor of 2 (i.e., the service time will become twice as long). Clearly, such a reduction will not affect the asymptotic bounds that were achieved using Lemmas \ref{lemma:sum_of_exp_bounded1}, and \ref{lemma:geom_lower_bound}. 
\end{proof}

\subsection{Proof of Lemma \ref{lemma:exp_server_instead_of_geom}}\label{app:C}

\textbf{Lemma \ref{lemma:exp_server_instead_of_geom}} (restated):
\textit{For any tree $G_v$ and $0 < p \le 1$:}
\begin{align*}
&\Pr\left(T(G_v, \text{Geom}(p))\leq t\right)\geq\Pr\left(T(G_v, \text{Exp}(p))\leq t\right),\\
&\text{ for all } t\ge0.
\end{align*}
\begin{proof}
We will prove this claim by showing that for each customer $c$ and for each queue $u$ on the unique path that $c$ traverses to the root, the probability that $c$ reaches $u$ before time $t$ is larger in $\mathcal{N}(G_v, \text{Geom}(p))$ than in $\mathcal{N}(G_v, \text{Exp}(p))$.

Consider a reverse topological order of the nodes in $G_v$, $v^1, v^2, \dots ,v^n=v$, i.e., for every node $v^i$, $1\le i < n$, the parent of $v^i$ is a node $v^j$ and $j > i$. For a node $v^i$ let $C^i$ be the set of customers that it needs to serve on their way to the root.
For a node $v^i$ and a customer $c \in C^i$ let $\mathcal{G}^i_c(t)$ denote the event that $c$ reached $v^i$ before time $t$ in 
$\mathcal{N}(G_v, \text{Geom}(p))$ and let $\mathcal{E}^i_c(t)$ be defined similarly for $\mathcal{N}(G_v, \text{Exp}(p))$.
We claim that for each $1 \le i \le n$, and each $c \in C^i$, $\Pr(\mathcal{G}^i_c(t)) \ge \Pr(\mathcal{E}^i_c(t))$, and the proof will be by induction on $i$.

\noindent \textbf{Induction basis}: $\Pr(\mathcal{G}^1_{v^1}(t)) \ge \Pr(\mathcal{E}^1_{v^1}(t))$. By definition $v_1$ is a leaf with one customer, itself, and no children, so $\Pr(\mathcal{G}^1_{v^1}(t)) = \Pr(\mathcal{E}^1_{v^1}(t)) =1$ for $t \ge 0$.

\noindent \textbf{Induction step}: Assume the claim is true for $1 \le i < n-1$ and we will prove it is true for $i+1$.
If $v^{i+1}$ is a leaf, then we are done since this is an identical case to the base case. Assume $v^{i+1}$ is not a leaf.
The case $c=v^{i+1}$ is trivial so consider $c \in C^{i+1}$ that is not $v^{i+1}$. Then $c$ must reach $v^{i+1}$ via one of its children, let it be $v^k$ where $k<i+1$. Then by the induction assumption $\Pr(\mathcal{G}^k_c(t')) \ge \Pr(\mathcal{E}^k_c(t'))$ for any $t'$,
and from Lemma \ref{lemma:geometric_as_exponential} for any $t$ the probability that a customer will be served by time $t$ is larger in $\mathcal{N}(G_v, \text{Geom}(p))$ than in $\mathcal{N}(G_v, \text{Exp}(p))$, so we have a faster arrival rate and a faster service rate and the claim follows.
\end{proof}

%

\subsection{Proof of Lemma \ref{lemma:geometric_as_exponential}}\label{app:D}
\textbf{Lemma \ref{lemma:geometric_as_exponential}} (restated):
\textit{Let $X$ be a geometric random variable with parameter $p$ and supported on the set $\left\{0,1,2,\ldots\right\}$, i.e., for $k\in\mathbb{Z^+}$: $\Pr \left(X=k\right)=(1-p)^{k}p$, and let $Y$ be an exponential random variable with parameter $p$. Then, for all $x\in \mathbb{R^+}$:
\begin{align}
\Pr \left(X\leq x \right)\geq \Pr \left(Y\leq x \right)=1-e^{-px},
\end{align}
i.e., a random variable $Y\sim\text{Exp}(p)$ stochastically dominates the random variable $X\sim\text{Geom}(p)$.}

\begin{proof}
For a geometric random variable $X$ with a success probability $p$ and supported on the set $\left\{0,1,2,3,...\right\}$:
$$\Pr\left(X>x\right)=(1-p)^{x+1} \text{, for }x\in \mathbb{Z^+}$$
and 
$$\Pr\left(X>x\right)=(1-p)^{\left\lfloor x\right\rfloor +1}\text{ , for }x\in \mathbb{R^+}.$$
So, for $x\in \mathbb{R^+}$, 
$$\Pr\left(X\leq x\right)=1-(1-p)^{\left\lfloor x\right\rfloor +1} \geq 1-(1-p)^{x}=1-e^{\ln(1-p)x}$$
and since $\ln(1-p)\leq -p$ we have:
$$\Pr\left(X\leq x\right)\ge 1-e^{-px}.$$
Hence, if $Y\sim\text{Exp}(p)$ we obtain:
$$\Pr\left(X\leq x\right)\ge1-e^{-px}=\Pr\left(Y\leq x\right),$$
i.e., random variable $Y\sim\text{Exp}(p)$ stochastically dominates the random variable $X\sim\text{Geom}(p)$.

\end{proof}

\subsection{Proof of Lemma \ref{lemma:sum_of_exp_bounded1}}\label{app:E}
\textbf{Lemma \ref{lemma:sum_of_exp_bounded1}} (restated):
\textit{Let $Y$ be the sum of $n$ independent and identically distributed exponential random variables (each with parameter $\mu>0$), and $\text{E}\left[Y\right]=\tfrac{n}{\mu}$.
Then, for $\alpha>1$:
\begin{align}
\Pr \left(Y < \alpha\text{E}\left[Y\right]\right) > 1-(2e^{-\alpha/2})^n.
\end{align}}

\begin{proof}
Let $Y=\sum_{i=1}^n X_i$, where $X_i$ are i.i.d. exponential random variables (each with parameter $\mu>0$).
The generating function of $X$ is given by:
$$G_{X}(s)=\text{E}\left[e^{sX}\right]=\int_{0}^{\infty}e^{sx}f_{X}(x)dx.$$
For any $s<\mu$:
$G_{X}(s)=\frac{\mu}{\mu-s}$.
Thus, the generating function of $Y$ (sum of independent $X_i$'s) for $s<\mu$:
$G_{Y}(s)=\left(G_{X}(s)\right)^n=\left(\frac{\mu}{\mu-s}\right)^n$.
Now, we will apply a Chernoff bound on $Y$. For $\mu>s\geq 0$:
\begin{align*}
\Pr \left(Y\geq \alpha\text{E}\left[Y\right]\right)&=\Pr \left(Y\geq \alpha\frac{n}{\mu}\right)\\
&=\Pr \left(e^{sY}\geq e^{s\cdot \alpha\frac{n}{\mu}}\right)\leq \frac{\text{E}\left[e^{sY}\right]}{e^{s\cdot \alpha\frac{n}{\mu}}}=\frac{G_{Y}(s)}{e^{s\cdot \alpha\frac{n}{\mu}}}.
\end{align*}
By letting $s=\mu/2$ we get:
$$\Pr \left(Y\geq \alpha\text{E}\left[Y\right]\right)\leq \left(\frac{\mu}{(\mu-\frac{\mu}{2})e^{\alpha\frac{\mu}{2\mu}}}\right)^n=\left(2e^{-\alpha/2}\right)^n$$
and thus:
$$\Pr \left(Y < \alpha\text{E}\left[Y\right]\right) > 1-\left(2e^{-\alpha/2}\right)^n.$$
\end{proof}

\subsection{Proof of the expectation result in Theorem \ref{thm:upper_bound_any_graph}}\label{app:F}
Theorem \ref{thm:upper_bound_any_graph} states that the expected stopping time of algebraic gossip on any graph is: $E[T]=O(\Delta n^2)$ timeslots.
\begin{proof}
First, we rewrite the high probability result of Eq. \ref{eq:t_upper} with $\alpha>1$:
$$\Pr\left(T \ge 4\alpha n^2\Delta\right)\le 2n(2e^{-\alpha/2})^n.$$
For a positive integer random variable $T$ holds: $\text{E}\left[T\right]=\sum_{i=1}^{\infty}\Pr(T\ge i)$. So, we have:
\begin{align}
\text{E}\left[T\right]&=\sum_{i=1}^{\infty}\Pr(T\ge i)
\\&=\sum_{i=1}^{8n^2\Delta-1}\Pr(T\ge i)+\sum_{i=8n^2\Delta}^{\infty}\Pr(T\ge i)\\
&\le 8n^2\Delta+\sum_{i=8n^2\Delta}^{\infty}\Pr(T\ge i)\\
&\le 8n^2\Delta+4n^2\Delta\sum_{\alpha=2}^{\infty}\Pr(T\ge 4\alpha n^2\Delta).
\end{align}
The last inequality is true since $\forall i\le j, \Pr(T\ge i)\ge \Pr(T\ge j)$ and thus we can replace all $\Pr(T\ge i)$ for $i\in[4\alpha n^2\Delta,...,4(\alpha+1)n^2\Delta-1]$ with $4n^2\Delta\times\Pr(T\ge 4\alpha n^2\Delta)$. Hence,
\begin{align}
\text{E}\left[T\right]&\le 8n^2\Delta+4n^2\Delta\sum_{\alpha=2}^{\infty}\Pr(T\ge 4\alpha n^2\Delta)\\
&\le 8n^2\Delta+4n^2\Delta\sum_{\alpha=2}^{\infty}2n(2e^{-\tfrac{\alpha}{2}})^n\\
&=8n^2\Delta+8n^3\Delta 2^n\sum_{\alpha=2}^{\infty}(e^{-n/2})^{\alpha}\\
&=8n^2\Delta+8n^3\Delta 2^n\frac{e^{-n}}{1-e^{-n/2}}\\
&=8n^2\Delta+ \frac{8n^3\Delta}{1-e^{-n/2}}\left(\frac{2}{e}\right)^n,\\
&\text{for }n>6:\\
&\le 8n^2\Delta+8n^2\Delta.
\end{align}
Thus: $E[T]=O(\Delta n^2)$, and  $E[R]=O(\Delta n)$.
\end{proof}

\subsection{Proof of Lemma \ref{lemma:geom_lower_bound}}\label{app:G}
\textbf{Lemma \ref{lemma:geom_lower_bound}} (restated):
\textit{Let $X$ be a sum of $m$ independent and identically distributed geometric random variables with parameter $p$, i.e., $X=\sum_{i=1}^{m} X_i$. Then, for any positive integer $k<m/p$
\begin{align}
\Pr \left(X > k\right) \geq 1-\left(\frac{m}{e^{\frac{m-kp}{m}}kp}\right)^{-m}.
\end{align}}

\begin{proof}
First, we will define $Y$ as the sum of $k$ independent Bernoulli random variables, i.e., $Y=\sum_{i=1}^{k}Y_i$, where $Y_i \sim Bernoulli(p)$.
Let us notice that: 
$$\Pr \left(X \leq k\right)= \Pr \left(Y \geq m\right)$$
The last is true since the event of observing at least $m$ successes in a sequence of $k$ Bernoulli trials implies that the sum of $m$ independent geometric random variables is no more than $k$. From the other side, if the sum of $m$ independent geometric random variables is no more than $k$ it implies that $m$ successes occurred not later than the $k$-th trial and thus $Y\geq m$.

Now we will use a Chernoff bound for the sum of independent Bernoulli random variables presented in \cite{1076315}:
For any $\delta > 0$ and $\mu = \text{E}\left[Y\right]$:
$$\Pr \left(Y \geq (1+\delta)\mu\right) < \left(\frac{e^{\delta}}{(1+\delta)^{1+\delta}}\right)^\mu.$$
Since $\mu=\text{E}\left[Y\right]=kp$, and by letting $\delta=\frac{m-kp}{kp}$, we obtain:
$$\Pr \left(Y \geq (1+\delta)\mu\right) = \Pr \left(Y \geq m \right) < \left(\frac{m}{e^{\frac{m-kp}{m}}kp}\right)^{-m}.$$
So:
$$\Pr \left(X \leq k\right) < \left(\frac{m}{e^{\frac{m-kp}{m}}kp}\right)^{-m},$$
and thus the result follows.
\end{proof}

\subsection{Proof of Claim \ref{clm-coupons}}\label{app:H}
\textbf{Claim \ref{clm-coupons}} (restated):
\textit{Let $X$ be the r.v. for the number of coupons needed to obtain $n$ distinct coupons (i.e., to obtain at least one coupon of each type), then:
$$E[X] = \Theta(n\log n)
\: \: \text{ and w.h.p. } \: \:
X = \Theta(n\log n).$$}
\begin{proof}
The first result (the expected value) and the upper bound w.h.p. are well known; see for example \cite{motwani95randomize,1076315}. We have not found a direct reference for the lower bound, namely that w.h.p. $X = \Omega(n \log n)$, so we give an outline here.
Let $\mathcal{E}_x$  denote the event that all $n$ different coupons have been collected after $X$ steps. Let $X=\sum_{i=1}^n X_i$ where $X_i$ is an r.v. that denotes the number of coupons of type $i$ collected. Clearly $X_i$'s are dependent.
To overcome this difficulty we will use Poisson approximation of the binomial random variable $X_i$ \cite{1076315}.
Consider $n$ Poisson independent random variables $Y_i$ ($i\in \left[1...n\right]$) with mean $\lambda=\frac{X}{n}$. Each variable represents the number of coupons of type $i$. Thus, the expected total number of coupons collected is $X$.
Let $\mathcal{E}_y$ denote the Poisson version of the event $\mathcal{E}_x$, i.e., that after collecting the different types of coupons independently with Poisson distribution with $\lambda$, we have at least one type of each coupon.
Since $Y_i$'s are i.i.d., we have $\Pr\left(\mathcal{E}_y\right)=\left(\Pr \left(Y_i\geq 1\right)\right)^n$.
It is clear that both $\Pr\left(\mathcal{E}_x\right)$ and $\Pr\left(\mathcal{E}_y\right)$ are monotonically increasing with $X$; therefore we can use the Poisson approximation that states that $\Pr\left(\mathcal{E}_x\right)\leq 2\Pr\left(\mathcal{E}_y\right)$ (\cite{1076315}, Corollary 5.11). Now, assume $X=n\ln n - n\ln\ln n$ and we have:
\begin{align*}
\Pr\left(\mathcal{E}_y\right)&=\left(1-e^{-(\ln n - \ln\ln n)}\right)^n =\left(1-\frac{\ln n}{n}\right)^n.
\end{align*}
Now we want to show that $\left(1-\frac{\ln n}{n}\right)^n\le \frac{1}{n}$. Let: $$z=n\left(1-\frac{\ln n}{n}\right)^n.$$
Then we obtain:
\begin{align*}
\ln z = \ln n +n\ln\left(1-\frac{\ln n}{n}\right).
\end{align*}
Using Taylor expansion we get:
\begin{align*}
\ln\left(1-\frac{\ln n}{n}\right)\le -\frac{\ln n}{n}.
\end{align*}
So:
\begin{align*}
\ln z \le \ln n-n\frac{\ln n}{n}=0.
\end{align*}
Since $\ln z \le 0$ we get that $z\le 1$ which yields: $\left(1-\frac{\ln n}{n}\right)^n\le \frac{1}{n}$.
So, $\Pr\left(\mathcal{E}_x\right)\leq 2\Pr\left(\mathcal{E}_y\right)\leq \frac{2}{n}$ and thus:
$$\Pr\left(X \geq n\ln n - n\ln\ln n\right)=1-\frac{2}{n}.$$
\end{proof}

\subsection{Proof of Claim \ref{lemma_sum_of_geom_vars}}\label{app:I}
\textbf{Claim \ref{lemma_sum_of_geom_vars}} (restated):
\textit{Let $X_i$ be independent geometric random variables with parameter $p$, and let $X=\sum_{i=1}^n X_i$.
For $p\geq \frac{1}{2}$, and $\alpha>1$:
\begin{align}
\Pr \left(X\geq 2n\alpha\right)\leq \left(2^{1.5-\alpha}\right)^n.
\end{align}}
\begin{proof}
In order to obtain this upper bound on the sum of $n$ independent geometric random variables we will use a Chernoff bound.
The generating function of a geometric random variable $X_i$ is given by:
$$G_{X_i}(t)=\text{E}\left[e^{tX_i}\right]=\frac{pe^t}{1-(1-p)e^t}\text{ }, \text{ where }t<-\ln(1-p).$$
The generating function of the sum of independent random variables is a multiplication of their generating functions. Thus:
$$G_{X}(t)=\text{E}\left[e^{t\sum_{i=1}^n X_i}\right]=\text{E}\left[e^{ntX_i}\right]=\left(\frac{pe^t}{1-(1-p)e^t}\right)^n.$$
Now, we will apply Markov's inequality to obtain an upper bound on $X$. For $t\geq 0$:
$$\Pr \left(X\geq 2n\alpha\right)=\Pr \left(e^{tX}\geq e^{t 2n\alpha}\right)\leq \frac{\text{E}\left[e^{tX}\right]}{e^{t 2n\alpha}}=\frac{G_{X}(t)}{e^{t 2n\alpha}}.$$
By letting $t=-0.5\ln(1-p)$ we get:
$$\Pr (X\geq 2n\alpha)\leq \left(\frac{(1-p)^{\alpha-0.5}p}{1-(1-p)^{0.5}}\right)^n.$$
It is clear that $\Pr \left(X\geq 2n\alpha\right)$ decreases when $p$ increases. Thus, to obtain an upper bound, we will substitute $p$ with its minimal value, i.e., $1/2$, and we get the result:
\begin{align*}
\Pr \left(X\geq 2n\alpha\right)&\leq \left(\frac{(1-0.5)^{\alpha-0.5}0.5}{1-(1-0.5)^{0.5}}\right)^n
\\&\leq \left(2\cdot0.5^{\alpha-0.5}\right)^n
\\&=\left(2^{1.5-\alpha}\right)^n.
\end{align*}
\end{proof}

\subsection{Proof of the expectation result in Theorem \ref{theorem_ex_is_better_than_push_pull}}\label{app:J}
The Eq. \eqref{s_expected} in Theorem \ref{theorem_ex_is_better_than_push_pull} states that: $E[R_2] \le 4n + 1$.

\begin{proof}
First, we rewrite the high probability result of Eq. \eqref{s_whp}  for $R_2$ with $\alpha>1$:
$$\Pr\left(R_2 \ge 2n\alpha \right)\le n(2^{1.5-\alpha})^n.$$
For a positive integer random variable $R_2$ holds: $\text{E}\left[R_2\right]=\sum_{i=1}^{\infty}\Pr(R_2\ge i)$. So, we have:
\begin{align*}
\text{E}\left[R_2\right]&=\sum_{i=1}^{\infty}\Pr(R_2\ge i)
\\&=\sum_{i=1}^{4n-1}\Pr(R_2\ge i)+\sum_{i=4n}^{\infty}\Pr(R_2\ge i)\\
&\le 4n+\sum_{i=4n}^{\infty}\Pr(R_2\ge i)\\
&\le 4n+2n\sum_{\alpha=2}^{\infty}\Pr(R_2\ge 2n\alpha).
\end{align*}
The last inequality is true since $\forall i\le j, \Pr(R_2\ge i)\ge \Pr(R_2\ge j)$ and thus we can replace all $\Pr(R_2\ge i)$ for $i\in[2n\alpha ,...,2n(\alpha+1)-1]$ with $2n\times\Pr(R_2\ge 2n\alpha)$. Hence,
\begin{align*}
\text{E}\left[R_2\right] &\le 4n+2n\sum_{\alpha=2}^{\infty}\Pr(R_2\ge 2n\alpha)\\
&\le 4n+2n\sum_{\alpha=2}^{\infty}n(2^{1.5-\alpha})^n\\
&=4n +2n^2\cdot 2^{1.5n}\sum_{\alpha=2}^{\infty}(2^{-n})^{\alpha}\\
&=4n +2n^2\cdot 2^{1.5n}\frac{2^{-2n}}{1-2^{-n}}\\
&=4n +\frac{2n^2\cdot 2^{0.5n}}{2^{n}-1},\\
&\text{and for } n>19:\\
&\le 4n + 1.
\end{align*}
\end{proof}

\subsection{Proof of Theorem \ref{thm:tree_of_queues}}\label{app:K}

For the proof of this theorem we need the following definitions, claims, and lemmas.

\subsubsection{Stochastic Dominance}\label{app:K-1}

\begin{dfn}[Stochastic dominance, stochastic ordering \cite{Hofstad08randomgraphs,GrimmettStirzaker:01}]
\label{dfn:stoch_dom}
We say that a random variable $X$ is stochastically less than or equal to a random variable $Y$ if and only if $\Pr(X\le t)\ge \Pr(Y\le t)$ for any $t\ge 0$, and such a relation is denoted as: $X\preceq Y$.
\end{dfn}

\begin{dfn}[Stochastic equivalence]
\label{dfn:stoch_equiv}
We say that a random variable $X$ is stochastically equivalent to a random variable $Y$ if and only if $\Pr(X\le t)= \Pr(Y\le t)$ for any $t\ge 0$, and such a relation is denoted as: $X\approx Y$.
\end{dfn}

The proof of the following two claims is omitted.

\begin{clm}
\label{clm:stoch_dom_max}
If for $i\in\{1,2\}$, $X_i\preceq Y_i$, $X_i$ are independent and $Y_i$ are independent, then: $\max_i{X_i}\preceq \max_i{Y_i}$.
\end{clm}


\begin{clm}
\label{clm:stoch_dom_sum}
If for $i\in\{1,2\}$, $X_i\preceq Y_i$, $X_i$ are independent and $Y_i$ are independent, then: $\sum_i{X_i}\preceq \sum_i{Y_i}$.
\end{clm}


\subsubsection{Later arrivals yield later departures}\label{app:K-2}

Consider an infinite FCFS queue with a single exponential server. We define $a_i$ as the time of arrival number $i$ to the queue, and $d_i$ as time of departure number $i$ from the queue. Let $X_i$ be the exponential random variable representing the service time of the arrival $i$. For all $i$, $X_i$'s are $i.i.d$.

Let $a_i$ be a sequence of $m$ arrival times to the queue, and $d_i$ be a sequence of $m$ departure times from the queue.

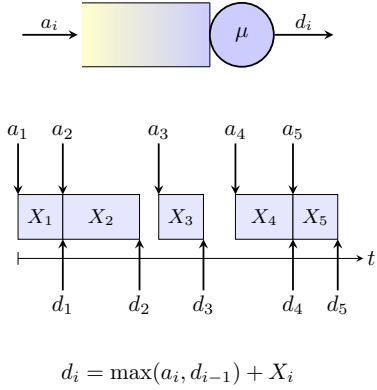
\begin{figure}[ht]
\centering
\scalebox{0.85}{\begin{tikzpicture}
[inner sep=0.6mm, place/.style={circle,draw=black,fill=blue!20,thick,minimum size=1cm},>=stealth]

\shade[left color=yellow!20,right color=blue!20] (0,0) rectangle (2,1);
\draw[black] (0,0) -- (2,0);
\draw[black] (0,1) -- (2,1);
\draw[black] (2,0) -- (2,1);
\node at (2.5,0.5) [place] (server)  {$\mu$};
\node at (0,0.5) [auto] (tail)  {\tiny{}};
\node at (4,0.5) [auto] (end)  {\tiny{}};
\node at (-1,0.5) [auto] (start)  {\tiny{}};

\path[->,thick] (start) edge node [above]{\small{$a_i$}} (tail);
\path[->,thick] (server) edge node [above]{\small{$d_i$}} (end);

\draw[black,fill=blue!10] (-1,-2) rectangle node[black] {\small{$X_1$}}(-0.3,-2.7) ;
\draw[black,fill=blue!10] (-0.3,-2) rectangle node[black] {\small{$X_2$}}(0.9,-2.7);
\draw[black,fill=blue!10] (1.2,-2) rectangle node[black] {\small{$X_3$}}(1.9,-2.7);
\draw[black,fill=blue!10] (2.4,-2) rectangle node[black] {\small{$X_4$}}(3.3,-2.7);
\draw[black,fill=blue!10] (3.3,-2) rectangle node[black] {\small{$X_5$}}(4,-2.7);

\draw[black,->] (-1,-3) -- (4.4,-3)node[black,right] {$t$};
\draw[black] (-1,-3.1) -- (-1,-2.9);

\draw[black,->,thick] (-0.3,-3.5) node[black,below] {$d_1$}-- (-0.3,-2.7);
\draw[black,->,thick] (0.9,-3.5) node[black,below] {$d_2$}-- (0.9,-2.7);
\draw[black,->,thick] (1.9,-3.5) node[black,below] {$d_3$}-- (1.9,-2.7);
\draw[black,->,thick] (3.3,-3.5) node[black,below] {$d_4$}-- (3.3,-2.7);
\draw[black,->,thick] (4,-3.5) node[black,below] {$d_5$}-- (4,-2.7);

\draw[black,->,thick] (-1,-1.2) node[black,above] {$a_1$}-- (-1,-2);
\draw[black,->,thick] (-0.3,-1.2) node[black,above] {$a_2$}-- (-0.3,-2);
\draw[black,->,thick] (1.2,-1.2) node[black,above] {$a_3$}-- (1.2,-2);
\draw[black,->,thick] (2.4,-1.2) node[black,above] {$a_4$}-- (2.4,-2);
\draw[black,->,thick] (3.3,-1.2) node[black,above] {$a_5$}-- (3.3,-2);

\node at (1.5,-4.8) [auto] (eq)  {$d_i=\max(a_i,d_{i-1})+X_i$};

\end{tikzpicture}}
\caption{Arrival and departure times.}
\label{fig:arrival_and_departure}
\end{figure}

\begin{lemma}\label{lemma:later_arrivals}
If the sequence $a_i$ is replaced with another sequence of $m$ arrivals -- $\hat{a_i}$, such that: $\hat{a}_i\succeq a_i$ $\forall i\in[1,...,m]$, then the resulting sequence of $m$ departures will be such that: $\hat{d}_i\succeq  d_i$ $\forall i\in[1,...,m]$. I.e., if every new arrival occurred, stochastically, at the same time or later than the old arrival, then every new departure from the queue will occur, stochastically, at the same time or later than the old departure.
\end{lemma}

\begin{proof}
The proof is by induction on the arrival index $j$, $j\in[1,...,m]$.
\begin{itemize}
\item Induction basis: $\hat{d}_1\succeq d_1$ follows since $d_1=a_1+X_1$, $\hat{d}_1=\hat{a}_1+X_1$, and $\hat{a}_1\succeq a_1$.
\item Induction assumption: $\forall i<j$ : $\hat{d}_i\succeq d_i$.
\item Induction step: we need to show that $\hat{d}_j\succeq d_j$.
\end{itemize}
If the $j$'s arrival occurred when the server was busy, then $d_j=d_{j-1}+X_j$. If the server was idle when the $j$'s arrival occurred, then $d_j=a_j+X_j$. Thus, we can write:
\begin{align}
d_j=\max(d_{j-1},a_j)+X_j, \\\text{and } \hat{d}_j=\max(\hat{d}_{j-1},\hat{a}_j)+X_j.
\end{align}
Since from induction assumption $\hat{d}_{j-1}\succeq d_{j-1}$, and $\hat{a}_j\succeq a_j$, using Claims \ref{clm:stoch_dom_max} and \ref{clm:stoch_dom_sum} we obtain $\hat{d}_j\succeq d_j$.
\end{proof}

\begin{proof}[Proof of Theorem \ref{thm:tree_of_queues}]

We denote the nodes of the queuing system $Q_n^{tree}$ as $Z_j^l$, where $l$ ($l\in[1,...,l_{\max}]$) is the level of the node in the tree, and $j$ is the node's index in level $l$. The root of the $Q_n^{tree}$ tree is the node $Z_1^1$. All servers in the $Q_n^{tree}$ network are ON all the time (work-conserving scheduling), i.e., servers work whenever they have customers to serve. There are no external arrivals to the system. Once a customer is serviced on level $l$, it enters the appropriate queue at the level $l-1$. When a customer is serviced by the root $Z_1^1$, it leaves the network.

Now, let us define the auxiliary queuing systems: $\hat{Q}_n^{tree}$ and $Q_{l_{\max}}^{line}$.

\begin{dfn}[Network $\hat{Q}_n^{tree}$]
\label{dfn:network_tn_hat}
$\hat{Q}_n^{tree}$ is the same network as $Q_n^{tree}$ with the following change in the servers' scheduling:

At any given moment, only one server at every level $l$ ($l\in[1,...,l_{\max}]$) is ON. Once a customer leaves level $l$,
a server that will be scheduled (turned ON) at level $l$, is the server that has in its queue a customer that has earliest arrival time to a queue at level $l$ among all the current customers at level $l$. If there are customers that initially reside at level $l$, they will be serviced in the order of their IDs (we assume for analysis that every customer has a unique identification number).
\end{dfn}

\begin{dfn}[Network of queues $Q_{l_{\max}}^{line}$]
\label{dfn:line_of_queues}
$Q_{l_{\max}}^{line}$ is the following modification of the network $Q_n^{tree}$ that results in a network of $l_{\max}$ queues arranged in a line topology.

For all $l\in[1,..,l_{\max}]$, we merge all the nodes at the level $l$ to a single node (a single queue with a single server). We name this single node at level $l$ as the first node in $Q_n^{tree}$ at level $l$, i.e., $Z_1^l$.
The customers that initially reside at level $l$ will be placed in a single queue in the order of their IDs.
This modification results in $Q_{l_{\max}}^{line}$ -- a network of $l_{\max}$ queues arranged in a line topology:
$Z_1^{l_{\max}}\rightarrow Z_1^{l_{\max}-1}\rightarrow\cdots\rightarrow Z_1^1$.
\end{dfn}

\begin{dfn}[Network of queues $\grave{Q}_{l_{\max}}^{line}$]
\label{dfn:line_of_queues_one_customer_back}
$\grave{Q}_{l_{\max}}^{line}$ -- is the same system as $Q_{l_{\max}}^{line}$ with the following modification.
We take the last customer at some node $Z_1^m$ ($m\in{[1,..,l_{\max}-1}]$) and place it at the head of the queue of node $Z_1^{m+1}$. I.e., we move one customer, one queue backward in the line of queues.
\end{dfn}

\begin{dfn}[Network of queues $\hat{Q}_{l_{\max}}^{line}$]
\label{dfn:line_of_queues_all_customer_back}
$\hat{Q}_{l_{\max}}^{line}$ -- is the same system as $Q_{l_{\max}}^{line}$ with the following modification.
We move all the customers to queue $Z^{l_{\max}}_1$. I.e., all the customers have to traverse now through all the $l_{\max}$ queues in the line.
\end{dfn}

We summarize the queuing systems defined above in short Table \ref{tab:queuing_systems}.

\renewcommand{\arraystretch}{1.2}
\begin{table}[h]
	\centering
	\scalebox{1}{
		\begin{tabular}{|c||m{7cm}|}
		\hline
			${Q}_n^{tree}$ & Original system of $n$ queues arranged in a tree topology. Fig. \ref{fig:three_network_systems} (a).\\\hline
			$\hat{Q}_n^{tree}$ & System of $n$ queues arranged in a tree topology. Only one server is active at each level at a given time. Fig. \ref{fig:three_network_systems} (b).\\\hline
			${Q}_{l_{\max}}^{line}$ & System of $l_{\max}$ queues arranged in a line topology. Fig. \ref{fig:three_network_systems} (c).\\\hline
			$\grave{Q}_{l_{\max}}^{line}$ & System of $l_{\max}$ queues arranged in a line topology. One customer is moved one queue backward.\\\hline
			$\hat{Q}_{l_{\max}}^{line}$ & System of $l_{\max}$ queues arranged in a line topology. All customers are moved backward to the queue $Z^{l_{\max}}_1$.\\\hline
		\end{tabular}
		}
	\caption{Queuing systems used in the proof.}
	\label{tab:queuing_systems}
\end{table}

The proof of Theorem \ref{thm:tree_of_queues} consists of showing the following relations between the stopping times of the queuing systems:
%
\begin{align*}
t(Q_n^{tree})&\preceq t(\hat{Q}_n^{tree})
\\&\approx t({Q}_{l_{\max}}^{line})
\\&\preceq t(\grave{Q}_{l_{\max}}^{line})
\\&\preceq t(\hat{Q}_{l_{\max}}^{line})=O(n/\mu).
\end{align*}

Stopping time of a queuing system $t(Q)$ is the time at which the last customer leaves the system (via the node $Z_1^1$).
In order to compare the stopping times of queuing systems, we define the following ordered set (or sequence) of departure time from a server $Z$ in a queuing system $Q$:
$d(Z,Q)=(d_1(Z,Q),d_2(Z,Q),...,d_i(Z,Q),...)$, where $d_i(Z,Q)$ is the time of the departure number $i$ from the node (server) $Z$.


\begin{figure*}[ht]
\centering
\input{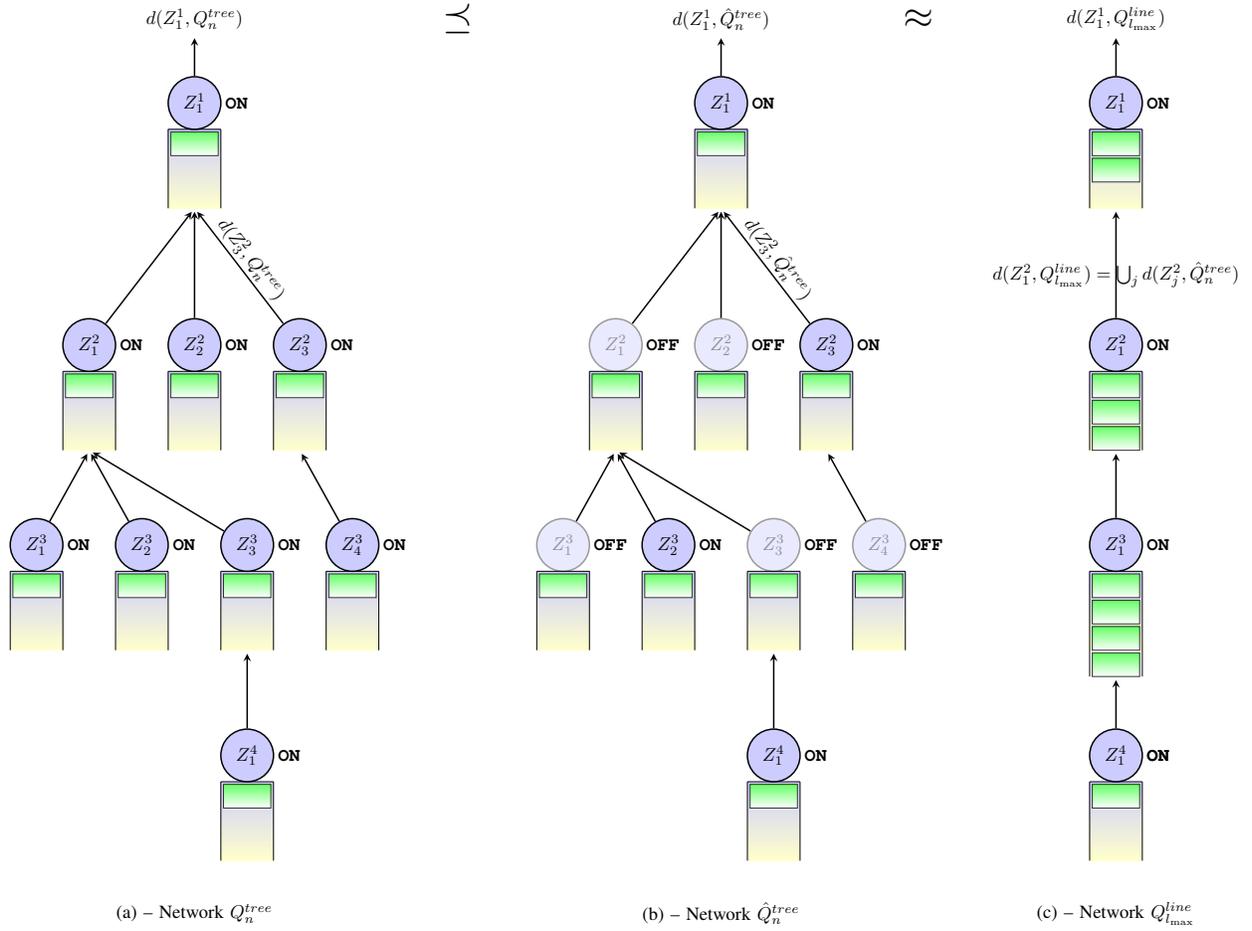}
\caption{(a) Network ${Q}_n^{tree}$, where all the servers work all the time.
(b) Network $\hat{Q}_n^{tree}$, where only one server at each level works at a given time.
(c) Network ${Q}_{l_{\max}}^{line}$.}
\label{fig:three_network_systems}
\end{figure*}

First, we want to show that the stopping time of $Q_n^{tree}$ is at most the stopping time of the system $\hat{Q}_n^{tree}$, i.e., $t(Q_n^{tree})\preceq t(\hat{Q}_n^{tree})$.

\begin{lemma}\label{lemma:idle_server}
In $\hat{Q}_n^{tree}$, every departure from the system (via $Z_1^1$) will occur, stochastically, at the same time or later than in ${Q}_n^{tree}$:
\begin{align}
{d}_i(Z_1^1,\hat{Q}_n^{tree})\succeq d_i(Z_1^1,{Q}_n^{tree}) \text{ } \forall i\in[1,...,n].
\end{align}
Thus, in $\hat{Q}_n^{tree}$, the last customer will leave the system, stochastically, at the same time or later than in ${Q}_n^{tree}$ or: $t({Q}_n^{tree})\preceq t(\hat{Q}_n^{tree})$.
\end{lemma}

\begin{proof}
The proof is by induction on the tree level $l$, $l\in[1,...,l_{\max}]$.
\begin{itemize}
\item Induction basis: $\forall i,j  \text{ : } {d}_i(Z_j^{l_{\max}},\hat{Q}_n^{tree})\succeq d_i(Z_j^{l_{\max}},{Q}_n^{tree})$. This is true since in $\hat{Q}_n^{tree}$, the nodes do not work all the time, and thus the departures will occur, stochastically, at the same time or later than in ${Q}_n^{tree}$. If there is a single node at level $l_{\max}$, in $\hat{Q}_n^{tree}$ it will be ON all the time as in ${Q}_n^{tree}$, and thus, the departures will occur, stochastically, at the same time in both systems.
\item Induction assumption: for all $l>m$ ($m\ge 1$), $\forall i,j  \text{ : } {d}_i(Z_j^{l},\hat{Q}_n^{tree})\succeq d_i(Z_j^{l},{Q}_n^{tree})$.
\item Induction step: we need to show that: $\forall i,j  \text{ : } {d}_i(Z_j^{m},\hat{Q}_n^{tree})\succeq d_i(Z_j^{m},{Q}_n^{tree})$.
\end{itemize}
By induction assumption, for $l=m+1$: $\forall i,j  \text{ : } {d}_i(Z_j^{m+1},\hat{Q}_n^{tree})\succeq d_i(Z_j^{m+1},{Q}_n^{tree})$.
Now let us take a look at the departures from a node $Z_j^{m}$. There are two cases: $Z_j^{m}$ is a leaf, and $Z_j^{m}$ is not a leaf.
If $Z_j^{m}$ is a leaf, we can use the same argument as in the induction basis: in $\hat{Q}_n^{tree}$, the node $Z_j^{m}$ does not work all the time, and thus the departures from it in $\hat{Q}_n^{tree}$ cannot occur earlier than in ${Q}_n^{tree}$.
If $Z_j^{m}$ is not a leaf, it has input/inputs of arrivals from the level $m+1$.
Since the arrivals from the level $m+1$ in $\hat{Q}_n^{tree}$ occur, stochastically, at the same time or later than in ${Q}_n^{tree}$ (by induction assumption), even if node $Z_j^{m}$ would work all the time (as in ${Q}_n^{tree}$), we would obtain from Lemma \ref{lemma:later_arrivals}: $\forall i,j \text{ : } {d}_i(Z_j^m,\hat{Q}_n^{tree})\succeq d_i(Z_j^m,{Q}_n^{tree})$. Moreover, in $\hat{Q}_n^{tree}$, node $Z_j^m$ does not work all the time (unless it is the only node at level $m$); thus the departure times in $\hat{Q}_n^{tree}$ can be even larger.
\end{proof}

\begin{lemma}\label{lemma:tree_as_line}
In ${Q}_{l_{\max}}^{line}$, every departure from the system (via $Z_1^1$) will occur, stochastically, at the same time as in $\hat{Q}_n^{tree}$.
Thus, in ${Q}_{l_{\max}}^{line}$, the last customer will leave the system, stochastically, at the same time as in $\hat{Q}_n^{tree}$.
\end{lemma}

\begin{proof}
Consider the two following facts regarding the network $\hat{Q}_n^{tree}$.
First, a customer entering level $l$ will be serviced after all the customers that arrived at level $l$ before it, are serviced.
Second, at any given moment, only one customer is being serviced at level $l$ (if there is at least one customer at the nodes $Z_j^l$). These facts are true due to the scheduling of the servers in $\hat{Q}_n^{tree}$ (Definition \ref{dfn:network_tn_hat}).

Clearly, the same facts are true for the network ${Q}_{l_{\max}}^{line}$. First, any customer entering level $l$ will be serviced after all the customers that arrived at level $l$ before it are serviced. Second, at any given moment, only one customer is being serviced at level $l$ (if there is at least one customer in node $Z_1^l$). These facts are true since in ${Q}_{l_{\max}}^{line}$, at every level, there is a single queue with a single server (Definition \ref{dfn:line_of_queues}).

So, the departure times of every customer from every level $l$ ($l\in[1,...,l_{\max}]$) are, stochastically, the same in both systems.
The departures from level $l=1$ are the departures from the node $Z_1^1$, and thus the lemma holds.
\end{proof}

Now we are going to move one customer one queue backward, and will show that the resulting system will have stochastically larger (or the same) stopping time.

\begin{lemma}
\label{lemma:move_one_customer_back}
Consider a network ${Q}_{l_{\max}}^{line}$. Let $m$ be a level index: $m\in{[1,..,l_{\max}-1]}$. We take the last customer at node $Z_1^m$ and place it at the head of the queue of node $Z_1^{m+1}$, and call the resulting network -- $\grave{Q}_{l_{\max}}^{line}$ (Fig. \ref{fig:moving_customers_in_line} (b)).
Then:
\begin{align}
{d}_i(Z_1^1,{Q}_{l_{\max}}^{line})\preceq {d}_i(Z_1^1,\grave{Q}_{l_{\max}}^{line}) \text{ } \forall i\in[1,...,n].
\end{align}
Thus, in $\grave{Q}_{l_{\max}}^{line}$, the last customer will leave the system, stochastically, at the same time or later than in ${Q}_{l_{\max}}^{line}$, or: $t({Q}_{l_{\max}}^{line})\preceq t(\grave{Q}_{l_{\max}}^{line})$.
\end{lemma}

\begin{proof}
We call the customer that was moved -- customer $c$.
Let us take a look at the times of arrivals to node $Z_1^{m}$ in ${Q}_{l_{\max}}^{line}$ and in $\grave{Q}_{l_{\max}}^{line}$.
Since customer $c$ is already located in the queue of $Z_1^m$ in ${Q}_{l_{\max}}^{line}$, its arrival time can be considered as $0$.
In $\grave{Q}_{l_{\max}}^{line}$, the arrival time of $c$ is at least $0$ (it should be serviced at $Z_1^{m+1}$ before arriving at $Z_1^m$). Each of the rest of the customers that should arrive at $Z_1^m$ will arrive in $\grave{Q}_{l_{\max}}^{line}$, stochastically, at the same time or later than in ${Q}_{l_{\max}}^{line}$, since in $\grave{Q}_{l_{\max}}^{line}$ the server $Z_1^{m+1}$ should first service the customer $c$, and only then start servicing the rest customers. Thus, ${d}_i(Z_1^{m+1},\grave{Q}_{l_{\max}}^{line})\succeq d_i(Z_1^{m+1},{Q}_{l_{\max}}^{line})$. Using Lemma \ref{lemma:later_arrivals} we obtain that: ${d}_i(Z_1^{m},\grave{Q}_{l_{\max}}^{line})\succeq d_i(Z_1^{m},{Q}_{l_{\max}}^{line})$. Iteratively applying Lemma \ref{lemma:later_arrivals} to the nodes $Z_1^l$, $l\in[m-1,...,1]$, we obtain the result: ${d}_i(Z_1^{1},\grave{Q}_{l_{\max}}^{line})\succeq d_i(Z_1^{1},{Q}_{l_{\max}}^{line})$.
\end{proof}

\begin{figure*}[ht]
\centering
\input{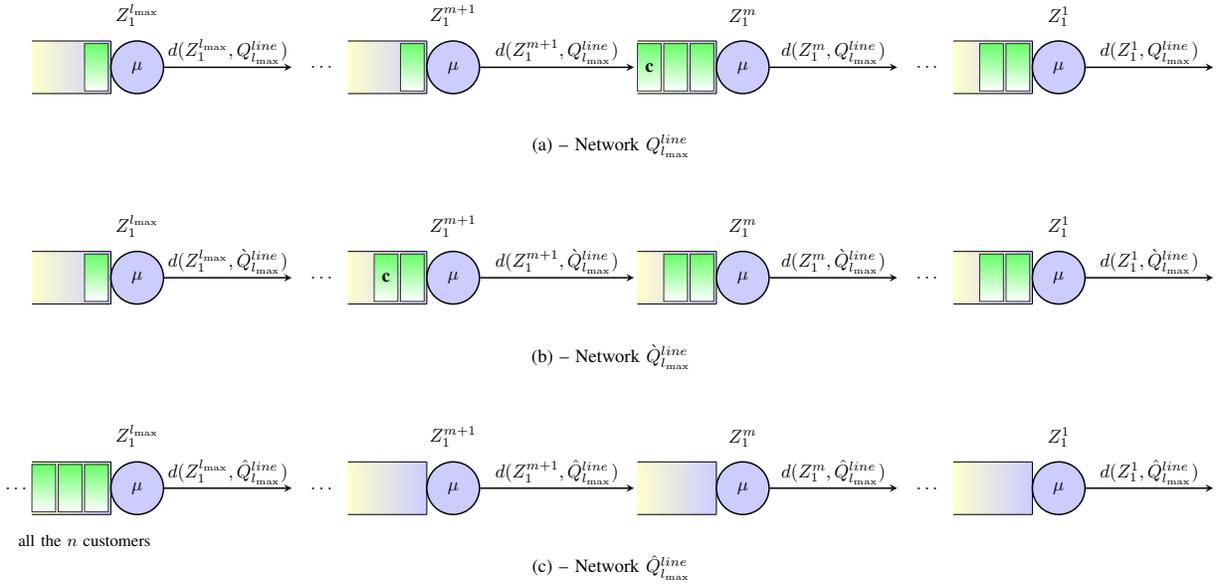}
\caption{(a) Network ${Q}_{l_{\max}}^{line}$.
(b) Network $\grave{Q}_n^{line}$, where one customer is moved one queue backward.
(c) Network $\hat{Q}_{l_{\max}}^{line}$, where all the customers are at the last queue.}
\label{fig:moving_customers_in_line}
\end{figure*}

\begin{corollary}
\label{corollary:take_all_customers_back}
Consider a network $\hat{Q}_{l_{\max}}^{line}$ (Definition \ref{dfn:line_of_queues_all_customer_back}) that is identical to the network ${Q}_{l_{\max}}^{line}$ with the following change. In $\hat{Q}_{l_{\max}}^{line}$, all $n$ customers are located at the node $Z_1^{l_{\max}}$ (Fig. \ref{fig:moving_customers_in_line} (c)).
Then:
\begin{align}
{d}_i(Z_1^1,{Q}_{l_{\max}}^{line})\preceq {d}_i(Z_1^1,\hat{Q}_{l_{\max}}^{line}) \text{ } \forall i\in[1,...,n].
\end{align}
Thus, in $\hat{Q}_{l_{\max}}^{line}$, the last customer will leave the system, stochastically, at the same time or later than in ${Q}_{l_{\max}}^{line}$, or: $t({Q}_{l_{\max}}^{line})\preceq t(\hat{Q}_{l_{\max}}^{line})$.
\end{corollary}

\begin{proof}
Given the network ${Q}_{l_{\max}}^{line}$, we take one customer from the tail of some queue (except the queue of node $Z_1^{l_{\max}}$) and place it at the head of the queue of the preceding node in the ${Q}_{l_{\max}}^{line}$. According to Lemma \ref{lemma:move_one_customer_back}, we get a network in which every customer leaves via $Z_1^1$, stochastically, not earlier than in ${Q}_{l_{\max}}^{line}$. Iteratively moving customers (one customer and one queue at a time) backwards we finally get the network $\hat{Q}_{l_{\max}}^{line}$ in which all $n$ customers are located at node $Z_1^{l_{\max}}$. Since at each step, according to Lemma \ref{lemma:move_one_customer_back}, the departure times from $Z_1^1$ could only get, stochastically, larger, the lemma holds.
\end{proof}

\begin{corollary}
\label{corollary:tree_is_slower_than_line}
The time it will take the last customer to leave the network of $n$ queues arranged in a \emph{tree} topology is, stochastically, the same or smaller than in the network of $n$ queues arranged in a \emph{line} topology where all $n$ customers are located at the farthest queue, i.e., $t({Q}_n^{tree})\preceq t(\hat{Q}_{l_{\max}}^{line})$.
\end{corollary}

\begin{proof}
This corollary is a direct consequence of Lemmas \ref{lemma:idle_server}, \ref{lemma:tree_as_line}, and the Corollary \ref{corollary:take_all_customers_back}.
\end{proof}

Now we are ready for the last step of the proof. We will find the stopping time of a system of queues arranged in a line topology and with all the customers located at the last queue.

\begin{lemma}
\label{lemma:line_stopping_time_with_k}
The time it will take for the last customer to leave system $\hat{Q}_{l_{\max}}^{line}$ ($l_{\max}$ queues arranged in a line topology) is $O(n/\mu)$ with high probability. Formally, for any $\alpha>1$:
\begin{align}
\Pr \left(t(\hat{Q}_{l_{\max}}^{line})< \alpha 4n/\mu\right) > 1-2(2e^{-\alpha/2})^n.
\end{align}
 
\end{lemma}

\begin{proof}
Initially, all the customers (from now we will call them \emph{real} customers) are located in the last ($Z_1^{l_{\max}}$) queue. We now take all the \emph{real} customers out of this queue and will make them enter the system (via $Z_1^{l_{\max}}$) from outside. We define the \emph{real} customers' arrivals as a Poisson process with rate $\lambda= \frac{\mu}{2}$. So, $\rho=\frac{\lambda}{\mu}=\frac{1}{2}<1$ for all the queues in the system. Clearly, such an assumption only increases the stopping time of the system (stopping time is the time until the last customer leaves the system).
According to Jackson's theorem, whose proof can be found in \cite{Chen2001Fundamentals}, there exists an equilibrium state. So, we need to ensure that the lengths of all queues at time $t=0$ are according to the equilibrium state probability distribution. We add \emph{dummy} customers to all the queues according to the stationary distribution. By adding additional \emph{dummy} customers
to the system, we make the \emph{real} customers wait longer in the queues, thus increasing the stopping time.

We will compute the stopping time $t(\hat{Q}_{l_{\max}}^{line})$ in two phases:
Let us denote this time as $t_1+t_2$, where $t_1$ is the time needed for the $n$'th customer to arrive at the first queue, and $t_2$ is the time needed for the $n$'th customer to pass through all the $l_{\max}$ queues in the system.

From Jackson's Theorem, it follows that the number of customers in each queue is independent, which implies that the random variables that represent the waiting times in each queue are independent.

The random variable $t_1$ is the sum of $n$ independent random variables distributed exponentially with parameter $\mu/2$. From Lemma \ref{lemma:waiting_distr} we obtain that $t_2$ is the sum of $l_{\max}$ independent random variables distributed exponentially with parameter $\mu-\lambda=\mu/2$ (Lemma \ref{lemma:waiting_distr}). Since $l_{max}\le n$, we can assume the worst case (for the upper bound of stopping time) $l_{max}=n$. Thus, we can view $t_2$ as the sum of $n$ independent random variables distributed exponentially with parameter $\mu/2$.
$\text{E}\left[t_1\right]=\sum_{i=1}^n 2/\mu=2n/\mu$, so, using Lemma \ref{lemma:sum_of_exp_bounded1}:
\begin{align}
\Pr \left(t_1 < \alpha\text{E}\left[t_1\right]\right) &> 1-(2e^{-\alpha/2})^n,
\\\Pr \left(t_1 < \alpha 2n/\mu\right) &> 1-(2e^{-\alpha/2})^n.
\end{align}
In a similar way we obtain:
\begin{align}
Pr \left(t_2 < \alpha 2n/\mu\right) &> 1-(2e^{-\alpha/2})^n.
\end{align}
$t(\hat{Q}_{l_{\max}}^{line})=t_1+t_2$; thus, using union bound:
\begin{align}
&\Pr \left(t_1+t_2 < \alpha 4n/\mu\right)>1-2(2e^{-\alpha/2})^n.
\\\notag &\text{and thus:}
\\ & t(\hat{Q}_{l_{\max}}^{line})=O(n/\mu) \text{ (for a constant $\alpha$)}
\\\notag&\text{ w.p. of at least } 1-2(2e^{-\alpha/2})^n.
\end{align}
\end{proof}

From Claim \ref{corollary:tree_is_slower_than_line} we obtain that $t({Q}_n^{tree})\preceq t(\hat{Q}_{l_{\max}}^{line})$ and thus: $t({Q}_n^{tree})< \alpha 4n/\mu$ w.p. of at least $1-2(2e^{-\alpha/2})^n$ for any $\alpha>1$, so the proof of Theorem \ref{thm:tree_of_queues} is completed.
\end{proof}


\end{document}